\documentclass[12pt]{article}

\usepackage{graphicx}
\usepackage{amsmath,amssymb}
\usepackage{comment}
\usepackage{setspace}
\usepackage{arydshln}
\usepackage{natbib}
\usepackage[english]{babel}
\usepackage{threeparttable}
\usepackage[a4paper]{geometry}
\usepackage{hyperref}
\usepackage{enumitem}
\hypersetup{
bookmarkstype=none,
colorlinks,
linkcolor=blue,
citecolor=black}
\allowdisplaybreaks

\usepackage{amsthm}
\newtheorem{theorem}{Theorem}[section]
\newtheorem{proposition}[theorem]{Proposition}
\newtheorem{corollary}[theorem]{Corollary}

\newtheorem{remark}[theorem]{Remark}
\newtheorem{definition}[theorem]{Definition}
\newtheorem{lemma}[theorem]{Lemma}
\newtheorem{example}[theorem]{Example}

\newcommand{\MX}{\mathcal{X}}
\newcommand{\MG}{\mathcal{G}}
\newcommand{\MF}{\mathcal{F}}
\newcommand{\MH}{\mathcal{H}}
\newcommand{\MY}{\{-1,1\}}
\newcommand{\MP}{\mathcal{P}}
\newcommand{\MR}{\mathcal{R}}
\newcommand{\Real}{\mathbb{R}}
\newcommand{\sign}{\mathrm{sign}}
\newcommand{\eps}{\varepsilon}
\newcommand{\indep}{\perp \!\!\! \perp}
\newcommand{\argmin}{\mathop{\rm arg~min}\limits}
\newcommand{\arginf}{\mathop{\rm arg~inf}\limits}

\makeatother
\begin{document}
\title{{Constrained Classification and Policy Learning}{\Large \thanks{%
We thank Ashesh Rambachan, J\"{o}rg Stoye, and Max Tabord-Meehan for valuable discussions and comments. We also thank participants at the 2021 Cowles Foundation Econometrics Conference, 2021 NASMES, 2021 SEA conference, and seminar participants at Bristol, CEMFI, Chicago, Cornell, CUHK, Glasgow, Northwestern, NYU, Penn State, SciencesPo, Syracuse, TAMU, UBC, UC-Berkeley, UC-Irvine, UC-Riverside, UPenn, UW Madison, and Z\"{u}rich for beneficial comments. The authors gratefully acknowledge financial support from ERC grant 715940,
the ESRC Centre for Microdata Methods and Practice (CeMMAP) (RES-589-28-0001),
Swiss NSF grant 192580,
and JSPS KAKENHI grant 22K20155.}}}
\author{Toru Kitagawa\thanks{ Department of Economics, Brown University. Email: toru\_kitagawa@brown.edu}, Shosei Sakaguchi\thanks{Faculty of Economics, University of Tokyo. Email: sakaguchi@e.u-tokyo.ac.jp}, and Aleksey Tetenov\thanks{Geneva School of Economics and Management, University of Geneva. Email:~aleksey.tetenov@unige.ch}}
\date{\today}
\maketitle
\begin{abstract}
Modern machine learning approaches to classification, including AdaBoost, support vector machines, and deep neural networks, utilize surrogate loss techniques to circumvent the computational complexity of minimizing empirical classification risk.
These techniques are also useful for causal policy learning problems, since estimation of individualized treatment rules can be cast as a weighted (cost-sensitive) classification problem. 
Consistency of the surrogate loss approaches studied in \citet{Zhang_2004} and \citet{Bartlett_et_al_2006} relies on the assumption of \textit{correct specification}, which means that the specified set of classifiers is rich enough to contain a first-best classifier.
This assumption is, however, less credible when interpretability or fairness constraints restrict the set of classifiers.
Consequently, the applicability of surrogate-loss-based algorithms in such second-best scenarios remains unknown. 
This paper studies the consistency of surrogate loss procedures under a constrained set of classifiers without assuming correct specification. 
We show that in settings where the constraint restricts the classifier's prediction set only, hinge losses (i.e., $\ell_1$-support vector machines) are the only surrogate losses that preserve consistency in second-best scenarios. 
If the constraint additionally restricts the functional form of the classifier, consistency of a surrogate loss approach is not guaranteed, even with hinge loss. 
We therefore characterize conditions on the constrained set of classifiers that can guarantee consistency of hinge-risk-minimizing classifiers. 
Exploiting our theoretical results, we develop robust and computationally attractive hinge-loss-based procedures for a monotone classification problem. 

\bigskip
\noindent
\textbf{Keywords:} Surrogate loss, support vector machine, monotone classification, fairness in machine learning, statistical treatment choice, personalized medicine
\end{abstract}

\newpage{}

\onehalfspacing

\section{Introduction}

Binary classification, the prediction of a binary dependent variable $Y \in \{-1, +1\}$ based upon covariate information $X \in \MX$, is one of the most fundamental problems in statistics and econometrics. Many modern machine learning algorithms build on statistically and computationally efficient classification algorithms, and their application has had a sizable impact on various fields of study and in society in general, e.g., pattern recognition, credit approval systems, personalized recommendation systems, to list but a few examples.
Since estimation of an optimal treatment assignment policy can be cast as a weighted (cost-sensitive) classification problem \citep{Zadrozny03}, methodological advances in the study of the classification problem apply to the causal problem of designing individualized treatment assignment policies. As the allocation of resources in both business and public policy settings has become more evidence-based and dependent upon algorithms, so too has there been increasingly active debate on how to make allocation algorithms respect societal preferences for interpretability and fairness \citep{Dwork_et_al_2012}. Understanding the theoretical performance guarantee and efficient implementation of classification algorithms under interpretability or fairness constraints is a problem of fundamental importance with a strong connection to real life. 

In the supervised binary classification problem, the typical objective is to learn a classification rule that minimizes the probability of false prediction. We denote the distribution of $(Y,X)$ by $P$, and a (non-randomized) classifier that predicts $Y \in \{-1, +1 \}$ based upon $\text{sign}(f(X))$ by $f: \MX \to \mathbb{R}$, where $\text{sign}(\alpha) = 1\{\alpha \geq 0\} - 1\{ \alpha < 0\}$. We denote the 0-level set of $f$ by $G_f \equiv \{x \in \MX : f(x) \geq 0 \} \subset \MX$, and refer to $G_f$ as the \textit{prediction set} of $f$. The goal is to learn a classifier that minimizes \textit{classification risk}:
\begin{equation} \label{eq:classification_risk}
R(f) \equiv P(\text{sign}(f(X)) \neq Y) = E_P[ 1 \{ Y \cdot \text{sign}( f(X) ) \leq 0 \} ]. 
\end{equation}
Given a training sample $\{ (Y_i,X_i) \sim_{iid} P : i=1,\dots,n \}$, the empirical risk minimization principle of \citet{Vapnik98} recommends estimating the optimal classifier by minimizing empirical classification risk,  
\begin{align}
&\hat{f} \in \arginf_{f \in \MF} \widehat{R}(f), \label{eq:ERM_optimization} \\
& \widehat{R}(f) \equiv \frac{1}{n} \sum_{i=1}^{n} 1\{ Y_i \cdot \mbox{sign}(f(X_i)) \leq 0 \},  \tag*{}
\end{align}
over a set of classifiers $\MF = \{f : \MX \to \mathbb{R} \}$. If the complexity of $\MF$ is properly constrained, the empirical-risk-minimizing (ERM) classifier $\hat{f}$ has statistically attractive properties including risk consistency and minimax rate optimality. See, for example, \citet{DGLbook96} and \citet{Lugosi02}.

Despite the desirable performance guarantee for the ERM classifier, the computational complexity of solving the optimization in (\ref{eq:ERM_optimization}) becomes a serious hurdle to practical implementation, especially when the dimension of covariates is moderate to large. To get around this issue, the existing literature has offered various alternatives to the ERM classifier, including support vector machines \citep{Cortes_Vapnik_1995}, AdaBoost \citep{Freund_Schapire_1997}, and neural networks. Focusing on optimization, each of these algorithms can be viewed as targeting the minimization of \textit{surrogate risk},
\begin{equation} \label{eq:surrogate_risk}
    R_{\phi}(f) \equiv E_P [\phi (Y f(X))], 
\end{equation}
where $\phi: \mathbb{R} \to \mathbb{R}$ is called the \textit{surrogate loss} function, a different specification of which corresponds to a different learning algorithm. Convex functions make for a desirable choice of surrogate loss function as, combined with some functional form specification for $f$, the  minimization problem for the empirical analogue of the surrogate risk in (\ref{eq:surrogate_risk}) is a convex optimization problem. This insight and the computational benefit that it yields has been pivotal to learning algorithms being able to handle large scale problems with high-dimensional features. 

Can surrogate risk minimization lead to an optimal classifier in terms of the original classification risk? 
The seminal works of \citet{Zhang_2004} and \citet{Bartlett_et_al_2006} provide theoretical justification for the use of surrogate losses by clarifying the conditions under which surrogate risk minimization also minimizes the original classification risk. 
A crucial assumption for this important result is \textit{correct specification} of the classifiers, requiring that the set of classifiers $\MF$ over which the surrogate risk is minimized contains a classifier that globally minimizes the original classification risk, i.e., a classifier that is identical to or performs as well as the Bayes classifier $f^{\ast}_{Bayes}(x) \equiv 2 P(Y=1|X=x) -1$ in terms of its classification risk. 

The credibility of the assumption of correct specification is, however, limited if the set of implementable classifiers is constrained exogenously, independently of any belief concerning the underlying data generating process. Such a situation is becoming more prevalent due to the increasing need for interpretability or fairness of classification algorithms. Given that $f$ determines the classification rule only through $G_f$, such constraints can be represented by shape restrictions on the prediction set of $f$, i.e., the class of feasible $f$ is represented by $\MF_{\MG} \equiv \{f \in \MF : G_f \in  \MG \}$, where $\MG$ is a restricted class of sets in $\MX$ satisfying the requirements for interpretability and fairness. To the best of our knowledge, how the validity of a surrogate loss approach is affected if $\MF_{\MG}$ misses the first-best classifier is not known.

The main contribution of this paper is to establish conditions under which a surrogate loss approach is valid without assuming correct specification. We first characterize the conditions on a surrogate loss under which minimizing the surrogate risk leads to a second-best rule (i.e., constrained optimum) in terms of the original classification risk. Specifically, we show that hinge losses $\phi_h(\alpha) = c \max\{ 0, 1- \alpha \}$, $c>0$, are the only surrogate losses that guarantee consistency of the surrogate risk minimization for a second-best classifier. An important implication of this result is that $\ell_1$-support vector machines are the only surrogate-loss-based methods that are robust to misspecification. 

The computational attractiveness of a surrogate loss approach crucially depends not only on the convexity of the surrogate loss function $\phi$ but also on the functional form restrictions on the classifier $f$ that lead to a convex $\MF$. We therefore investigate how additional constraints on $f$ beyond the requirement that $G_f \in \MG$ can affect the consistency of the hinge risk minimization. As a second contribution of this paper, we characterize a simple-to-check sufficient condition for consistency of the hinge risk minimization in terms of the additional functional form restrictions we can impose on $\MF_{\MG}$. We call a subset of $\MF_{\MG}$ that satisfies this sufficient condition a \textit{classification-preserving reduction} of $\MF_{\MG}$. 

Exploiting our main theoretical results, we develop novel procedures for monotone classification. In monotone classification, prediction sets are constrained to 
\begin{equation}
\MG_M \equiv \{G \subset \MX : x \in G  \Rightarrow x' \in G \ \forall x' \geq x \}, \notag
\end{equation}
where $x' \geq x$ is an element-wise weak inequality. Since $\MG_M$ coincides with the class of prediction sets spanned by the class of non-decreasing bounded functions $\MF_M \equiv \{ f: f \ \text{is non-decreasing in }x, \  -1 \leq f \leq  1 \}$, hinge-loss-based estimation for monotone classification can be performed by solving
\begin{align}
& \hat{f}_M \in \arginf_{f \in \MF_M} \widehat{R}_{\phi_h} (f), \label{eq:monotone_empriical_surrogate_risk} \\
& \widehat{R}_{\phi_h} (f) \equiv \frac{1}{n} \sum_{i=1}^n \phi_h (y_i f(x_i)). \notag
\end{align}
We show that the class of monotone classifiers $\MF_M$ is a constrained classification-preserving reduction of $\MF_{\MG_M}$, guaranteeing consistency of the hinge-risk-minimizing classifier $\hat{f}_M$. 
Furthermore, we show that convexity of $\MF_{M}$ reduces the optimization of (\ref{eq:monotone_empriical_surrogate_risk}) to a finite dimensional linear programming problem and hence delivers significant computational gains relative to minimization of the original empirical classification risk. Adopting this approach, the application of our main theorems guarantees
\begin{equation}
    R(\hat{f}_M) - \inf_{f \in \MF_M}R(f) \xrightarrow[p]{} 0, \notag
\end{equation}
as $n \to \infty$, and this convergence is valid regardless of whether $\MF_M$ attains the first-best risk, i.e., $\inf_{f \in \MF_M} R(f) = \inf_{f \in \bar{\MF}}R(f)$, or not, where $\bar{\MF}$ is the class of measurable functions $f: \MX \to \mathbb{R}$. We also derive the uniform upper bound of the mean of $R(\hat{f}_M) - \inf_{f \in \MF_M}R(f)$ to characterize the regret convergence rate attained by $\hat{f}_M$. 

\subsection{Connection and contributions to causal policy learning} \label{sec:Connection and contributions to causal policy learning}
For simplicity of exposition, this paper mainly focuses on the prototypical setting of binary classification. The main theoretical results can easily be extended to weighted (cost-sensitive) classification, where the canonical representation of the population risk criterion is given by
\begin{equation}
R^{\omega}(f) \equiv E_P[\omega \cdot 1 \{ Y \cdot \text{sign}( f(X) ) \leq 0 \} ]. \label{eq:weighted_classification_risk}
\end{equation}
Here, $\omega$ is a non-negative random variable defining the cost of misclassifying $Y$ that typically depends on $(Y,X)$. The cost of misclassification $\omega$ may represent the decision-maker's economic cost \citep{LieliWhite2010} or welfare weights over the individuals to be classified, as considered in \citet{Rambachan_2020} and \citet{Babii_et_al_2020}. The surrogate risk for weighted classification can be defined similarly to (\ref{eq:surrogate_risk}), as
\begin{equation}
    R_{\phi}^{\omega}(f) = E_P[\omega \cdot \phi(Y f(X)) ]. \label{eq:weighted_surrogate_risk}
\end{equation}

As discussed in \citet{KT18}, there are fundamental conceptual differences between the prediction problem of classification and the causal problem of treatment choice.
Nevertheless, if the training sample is obtained from a randomized control trial (RCT) or an observational study satisfying unconfoundedness (selection on observables), we can view minimization of the weighted classification risk in (\ref{eq:weighted_classification_risk}) as being equivalent to the maximization of the additive welfare criterion commonly specified in treatment choice problems. To see this equivalence, let $\{ (Z_i,D_i,X_i) : i=1, \dots, n \}$ be an independent and identically distributed RCT sample of $n$ experimental subjects, where $Z_i \in \mathbb{R}$ is subject $i$'s observed outcome, $D_i \in \{-1, +1\}$ is an indicator for his assigned treatment, and $X_i \in \MX$ is a vector of pretreatment covariates, and let $(Z_i(d): d \in \{-1,+1\})$ be $i$'s potential outcomes satisfying $Z_i= Z_i(+1) \cdot 1\{D_i = +1 \} + Z_i(-1) \cdot 1\{D_i = -1 \}$. We denote the propensity score in the RCT sample by $e(x) \equiv P(D = +1|X=x)$ and assume that $e(x)$ is bounded away from 0 and 1 for all $x \in \MX$. We denote the joint distribution of $(Z_i(+1),Z_i(-1),D_i,X_i)$ by $P$ and assume $P$ satisfies unconfoundedness, $(Z(+1), Z(-1)) \indep D | X$.

Similar to our consideration of classification, we represent a (non-randomized) treatment assignment rule by the sign of  $f: \MX \to \mathbb{R}$ --i.e., the 0-level set $G_f = \{x \in \MX : f(x) \geq 0 \} \subset \MX$ specifies the subgroup of the population assigned to the treatment $+1$. Following \citet{Manski2004}, we consider evaluating the welfare performance of the assignment policy $f$ by the average outcomes attained under its associated assignment rule:
\begin{equation}
    W(f) \equiv E_P \left[ Z(+1) \cdot 1\{ X \in G_f \} + Z(-1) \cdot 1\{ X \notin G_f \} \right] \nonumber
\end{equation}
Relying on unconfoundedness of the experimental data and employing the inverse propensity score weighting technique, we can express this welfare in terms of the observable variables as\footnote{\citet{KL21} makes use of this transformation of the welfare objective function to develop an Adaboost algorithm for treatment choice.}
\begin{align}
W(f) &= E_P \left[ \frac{Z}{D e(X) + (1- D)/2 } \cdot 1\{ D = \text{sign}(f(X)) \} \right] \notag \\
     & = E_P \left[ \max \left\{0,\frac{Z}{D e(X) + (1- D)/2 } \right\} \right]  - E_P \left[ \omega_p \cdot 1\{\sign(Z) \cdot D \cdot \text{sign}(f(X)) \leq 0 \} \right], \mspace{15mu}  \label{eq:policy_learning_welfare}
     \intertext{where}
     \omega_p & \equiv \frac{|Z|}{D e(X) + (1- D)/2 } \geq 0. \notag
\end{align}
Provided that the first moment of $ \omega_p$ is finite, maximization of $W(f)$ is equivalent to minimization of the weighted classification risk $R^{\omega}(f)$ defined in (\ref{eq:weighted_classification_risk}) with $\omega = \omega_p$ and $Y=\sign(Z)\cdot D$. As a result, optimal treatment assignment rules can be viewed as optimal classifiers for $D$ in terms of weighted classification risk. This equivalence also holds for other methods of policy learning, such as the offset-tree learning of \citet{BeygelzimerLangford09} and the doubly-robust approaches of \citet{SJ15} and \citet{AW17}, which correspond to different ways of constructing or estimating the weighting term $\omega_p$.  

Due to its equivalence to weighted classification, a surrogate loss approach to policy learning proceeds by minimizing the empirical analogue of (\ref{eq:weighted_surrogate_risk}) with $\omega = \omega_p$ and $Y=\sign(Z)\cdot D$. 
Section \ref{sec:Extension to individualized treatment rules} of this paper shows that our main theoretical results established for constrained binary classification carry over to the setting of policy learning in which feasible treatment assignment policies are constrained exogenously due to fairness and legislative considerations. This paper therefore offers valuable and novel contributions to current research and public debate regarding how to make use of machine learning algorithms to design individualized policies. If treatment assignment rules are constrained to be monotone, our concrete proposals for monotone classification algorithms can be applied to policy learning, which yields significant gains in computational efficiency relative to the mixed integer programming approaches considered in \citet{KT18} and \citet{MT17}. 

\subsection{Related literature}

This paper is closely related to the literature of consistency and performance guarantees for surrogate risk minimization. Notable works in this literature include \cite{Mannor_et_al_2003}, \cite{Jiang_2004}, \cite{Lugosi_Vayatis_2004} , \cite{Zhang_2004}, \citet{Steinwart_2005, Steinwart_2007}, \cite{Bartlett_et_al_2006}, \cite{Nguyen_et_al_2009}, and \cite{Scott_2012}.
Under the assumption of correct specification, \cite{Zhang_2004} and \cite{Bartlett_et_al_2006} derive quantitative relationships between excess classification risk and excess surrogate risk, and then provide general conditions for surrogate risk minimization to achieve risk consistency. \cite{Bartlett_et_al_2006} show that the classification-calibration property of surrogate loss, defined in Section \ref{sec:calibration of MG-constrained classification} below, guarantees risk consistency.
\cite{Zhang_2004} and \cite{Bartlett_et_al_2006} show that many commonly used surrogate loss functions, including hinge loss, exponential loss, and truncated quadratic loss, satisfy the conditions needed for risk consistency.
In a classification problem different from ours, where a pair comprising a quantizer and a classifier is chosen, \cite{Nguyen_et_al_2009} study sufficient and necessary conditions for surrogate risk minimization to yield risk consistency.
\cite{Nguyen_et_al_2009} show that only hinge loss functions satisfy the conditions required for risk consistency in their problem. 
\cite{Liu_et_al_2025} study identification for an optimal classifier and inference for optimal risk value with surrogate risk minimization.
Correct specification of the set of classifiers is an essential condition for consistency in all of the surrogate risk minimization approaches studied in the literature. 
The key contribution of our paper is to relax the assumption of correct specification and to clarify the conditions that are required for the surrogate loss function to yield a consistent surrogate risk minimization procedure.

Relaxing the assumption of correct specification connects this paper to classification problems with exogenous constraints. Such problems are studied in machine learning and statistics, and include interpretable classification \citep[e.g.,][]{Zeng_et_al_2017,Zhang_et_al_2018}, fair classification \citep[e.g.,][]{Dwork_et_al_2012}, and monotone classification \citep[e.g.,][]{Cano_et_al_2019}. Some works in the existing literature adopt a surrogate loss approach. \cite{Donini_et_al_2018} use the $\ell_1$-support vector machine in fair classification, where the hinge risk minimization is subject to a statistical fairness constraint. \cite{Chen_Li_2014} use the $\ell_1$-support vector machine with a monotonicity constraint, which constrains the class of feasible classifiers to a class of certain monotone functions. However, neither paper shows consistency of their hinge risk minimization procedures in terms of classification risk.
\cite{Agarwal_et_al_2018} propose an approach to reduce fairness constrained classification problems to weighted classification. Their reduction can accommodate various fairness constraints proposed in the machine learning literature and the monotonicity constraint of this paper if $X$ is discrete. Our consistency results on surrogate-risk-minimizing classifiers can apply to an arbitrary class $\MG$ regardless of whether their reduction applies or not. 

Focusing on optimization, ERM classification and maximum score estimation \citep{Manski1975,ManskiThompson1989} share the same objective function. 
\cite{Horowitz_1992} proposes smooth maximum score estimation, where kernel smoothing is performed on the 0-1 loss to obtain a differentiable objective function. However, the smoothed objective function remains non-convex and does not offer the computational gains that the surrogate risk minimization approach with convex surrogates can deliver.

This paper also contributes to a growing literature on statistical treatment rules in econometrics, including \cite{Manski2004}, \cite{Dehejia2005}, \cite{HiranoPorter2009}, \citet{Stoye2009,Stoye2012}, \cite{Chamberlain2011}, \cite{BhattacharyaDupas2012}, \cite{Tetenov2012}, \cite{Kasy2018}, \cite{KT18, KT21}, \cite{Viviano_2019}, \cite{AW17}, \cite{MT17}, \cite{Sakaguchi_2021}, \cite{KW23}, among others. As discussed above, the policy learning methods of \cite{KT18}, \cite{AW17}, and \cite{MT17} build on the similarity between empirical welfare maximizing treatment choice and ERM classification.
\cite{MT17} propose penalization methods to control the complexity of treatment assignment rules, and derive relevant finite sample upper bounds on the regret of the estimated treatment rules. 
\cite{AW17} apply doubly-robust estimators to estimate the weight $\omega$ in (\ref{eq:weighted_surrogate_risk}), and show that an $1/\sqrt{n}$-upper bound on regret can also be achieved in the observational study setting.  
These works optimize an empirical welfare objective involving an indicator loss function. As a result, the practical implementation of such methods is sometimes discouraging, especially when the sample size or number of covariates is moderate to large.

Estimation of individualized treatment rules is a topic of active research in other fields including medical statistics, machine learning, and computer science. Notable works in these fields include \cite{Zadrozny03}, \cite{BeygelzimerLangford09}, \cite{Qian_Murphy_2011}, \cite{Zhao2012JASA}, \cite{SJ15}, \cite{Zhao_et_al_2015}, and \cite{Kallus_2020}, among others. \cite{Zhao2012JASA} propose using $\ell_{1}$-support vector machines to solve the weighted classification with individualized treatment choice problem, and show risk consistency. They specify a rich class of treatment choice rules that is a reproducing kernel Hilbert space, and assume correct specification. 

\section{Illustrative examples and roadmap}
\label{sec:examples}

We start with a minimal numerical example that will provide a roadmap to the questions
we consider and to our theoretical results.
Throughout the paper we will denote the marginal distribution of $X$ by $P_X$ and
the conditional probability of $Y=+1$ given $X$ by
\begin{equation}
\label{eq:definition of eta}
\eta(X) \equiv P(Y=+1|X).
\end{equation}

For our illustrative example, we will consider probability distribution $P$ on $(Y,X)$ 
with the covariate $X$ taking three values $\MX = \{0,1,2\}$.
This distribution is characterized by $P_X$ and $\eta(X)$ as follows:
\begin{equation}
\label{eq:example 1 distribution}
\begin{array}{r|c|c|c}
x = & 0 & 1 &  2\\ \hline
P_{X}(x) = & 1/6 & 1/2 & 1/3 \\ 
\eta(x) = & 0.9 & 0.3 & 0.8\\
\end{array}
\end{equation}
A classifier $f:\MX \rightarrow \Real$ minimizes classification risk (\ref{eq:classification_risk})
under $P$ if and only if $f(0) \geq 0$, $f(1) < 0$, and $f(2) \geq 0$.

First, consider the baseline problem of \emph{unconstrained surrogate risk minimization},
in which the surrogate risk $R_\phi(f)$ is minimized over all classifiers
$\MF_{all} \equiv \{f : \MX \rightarrow [-1,1]\}$.
Figure \ref{fig:unconstrained example} displays classifiers minimizing classification risk
and surrogate risks based on hinge loss $\phi_h(\alpha) = \max(0,1-\alpha)$
and exponential loss $\phi_e(\alpha) = \exp(-\alpha)$.
Both surrogate risk minimizers have the same sign  as the minimizer of classification risk for every $X$, hence both also minimize classification risk.

This is unsurprising: both of these loss functions are \emph{classification-calibrated}.
As \citet{Bartlett_et_al_2006} have shown, unconstrained risk minimization with a classification-calibrated
surrogate loss function yields a classifier that also minimizes classification risk. 

\bigskip
\begin{figure}[ht]
\caption{Unconstrained minimization of classification and surrogate risks}
\begin{center}
\begin{tabular}{c}
\includegraphics[scale=0.6]{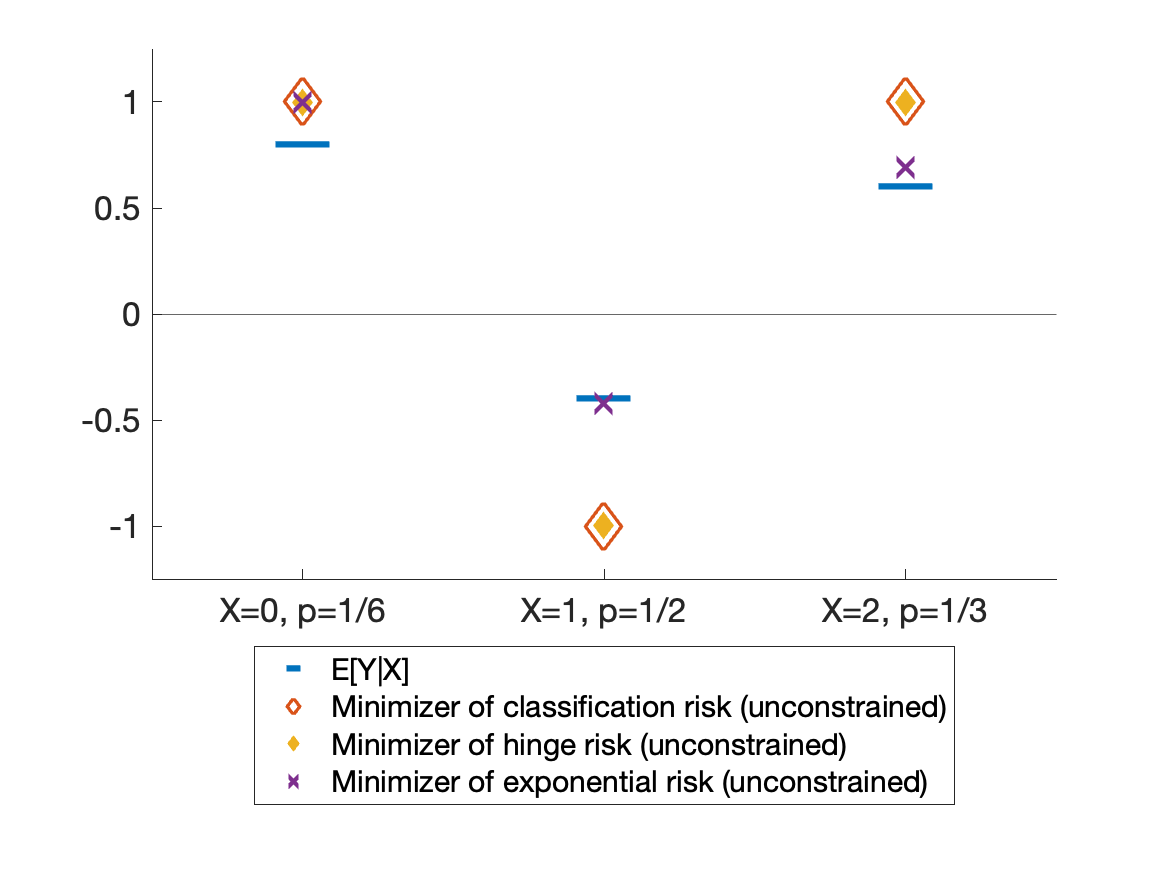}
\label{fig:unconstrained example}
\end{tabular}
\end{center}
\end{figure}

Our main interest is to investigate the properties of surrogate risks in constrained
classification problems where we want to obtain a prediction set $G_f$ that belongs
to an exogenously pre-specified class of sets $\MG \subset 2^{\MX}$.
For our example, consider a monotonicity restriction on prediction sets with
$\MG  = \{\emptyset, \{2\}, \{1,2\},\{0,1,2\}\}$. For distribution $P$ in our example, 
this restriction rules out the first-best classifier
$(f^{\ast}_{Bayes})$ with prediction set $\{0,2\}$.

In Sections \ref{sec:constrained classification and surrogate loss approach} and \ref{sec:calibration of MG-constrained classification}, we consider minimization of surrogate risk over the largest set of classifiers
whose prediction sets belong to $\MG$, denoted by
$\MF_{\MG} \equiv \{ f : G_f \in \MG, f(\cdot) \in [-1,1]\}$.
In our current example, this is the set of all classifiers
$f \in [-1,1]^3$ with the restriction 
$\sign(f(0)) \leq \sign(f(1)) \leq \sign(f(2))$.
The minimizers of classification and surrogate risks over this
set of classifiers are shown in Figure \ref{fig:G-constrained example}.
With the monotonicity constraint, classification risk is lowest for classifiers
with $f(0) < 0, f(1)<0, f(2) \geq 0$. Minimizing hinge risk yields such
a classifier, whereas minimizing exponential risk yields a suboptimal classifier
with $f(0) \geq 0, f(1) \geq 0, f(2) \geq 0$.
This illustrates a general result established in 
Section \ref{sec:calibration of MG-constrained classification}:
minimizing hinge risk over $\MF_{\MG}$ always yields a solution
also minimizing classification risk, whereas minimizing other common
surrogate risks (including exponential) generally does not.

\bigskip
\begin{figure}[ht]
\caption{$\MG$-constrained minimization of classification and surrogate risks}
\begin{center}
\begin{tabular}{c}
\includegraphics[scale=0.6]{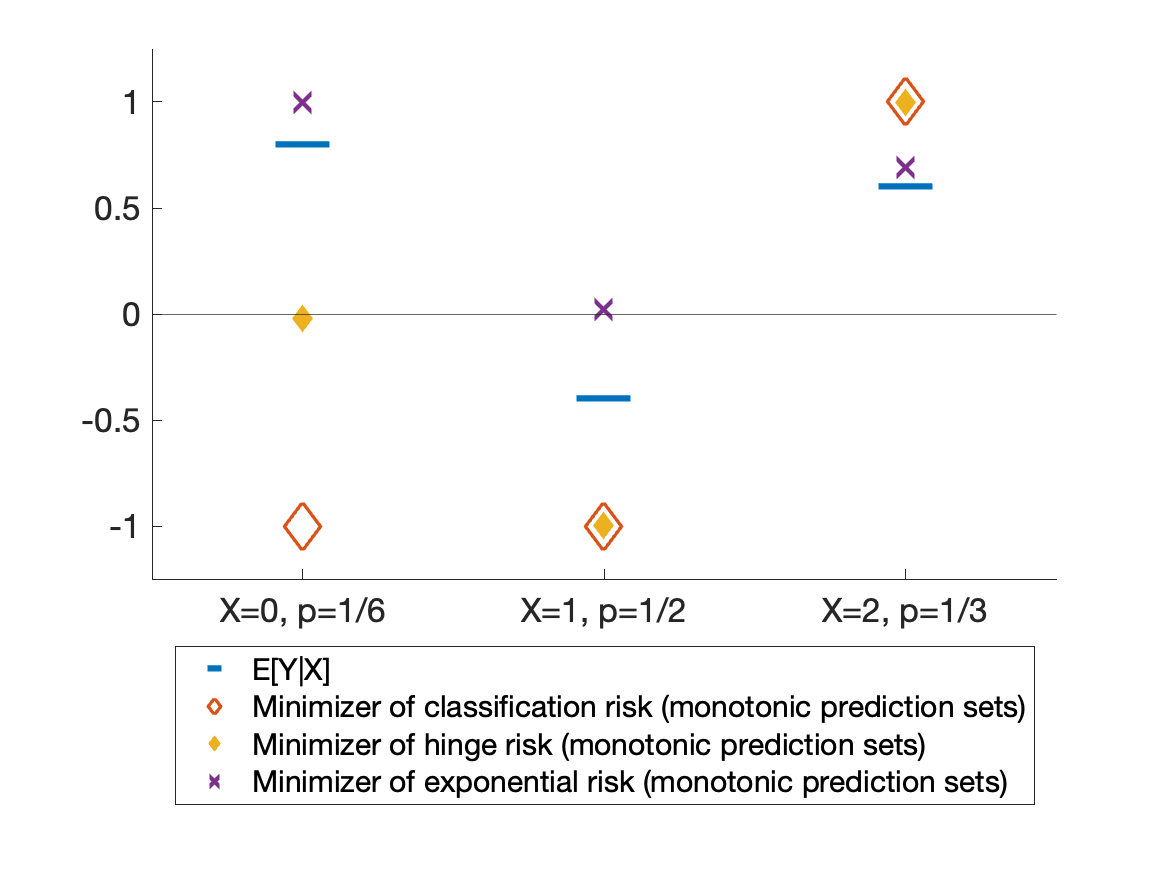}
\label{fig:G-constrained example}
\end{tabular}
\end{center}
\end{figure}

The set of classifiers $\MF_{\MG}$ is non-convex,
hence minimizing a convex risk over it does not yield a clear computational 
advantage. We would like to search over a convex subset of $\widetilde{\MF}_{\MG} \subseteq \MF_{\MG}$.
Section \ref{sec:classification preserving reduction} derives sufficient conditions
on $\widetilde{\MF}_{\MG}$ under which minimizing hinge risk yields a classifier
minimizing classification risk with the constraint that its prediction set must belong to $\MG$.
The restriction $\MG  = \{\emptyset, \{2\}, \{1,2\},\{0,1,2\}\}$ in our current example is a special
case of a monotone classification problem covered in Section \ref{sec:classification preserving reduction}.
The optimal classifier can be obtained in our example by minimizing hinge risk over
a convex set $\widetilde{\MF}_{\MG} = \{f : f(0) \leq f(1) \leq f(2), f \in [-1,1]^3\}$.

Figure \ref{fig:monotonic example} shows that the minimizer of hinge risk over $\widetilde{\MF}_{\MG}$
(yellow diamonds) coincides with the minimizer of classification risk under the monotonicity
constraint on prediction sets. Figure \ref{fig:monotonic example} also shows the minimizer of hinge risk over the set of classifiers that are linear non-decreasing functions (purple line). While the set of linear non-decreasing classifiers is convex and contains classifiers with prediction sets for every $G \in \MG$, it does
not satisfy the sufficient conditions we derive in Section \ref{sec:classification preserving reduction}.
The minimizer of hinge risk over linear non-decreasing $f(\cdot)$ in this case is suboptimal.

\bigskip
\begin{figure}[ht]
\caption{Minimization of classification and surrogate risks over monotone classes of functions}
\begin{center}
\begin{tabular}{c}
\includegraphics[scale=0.6]{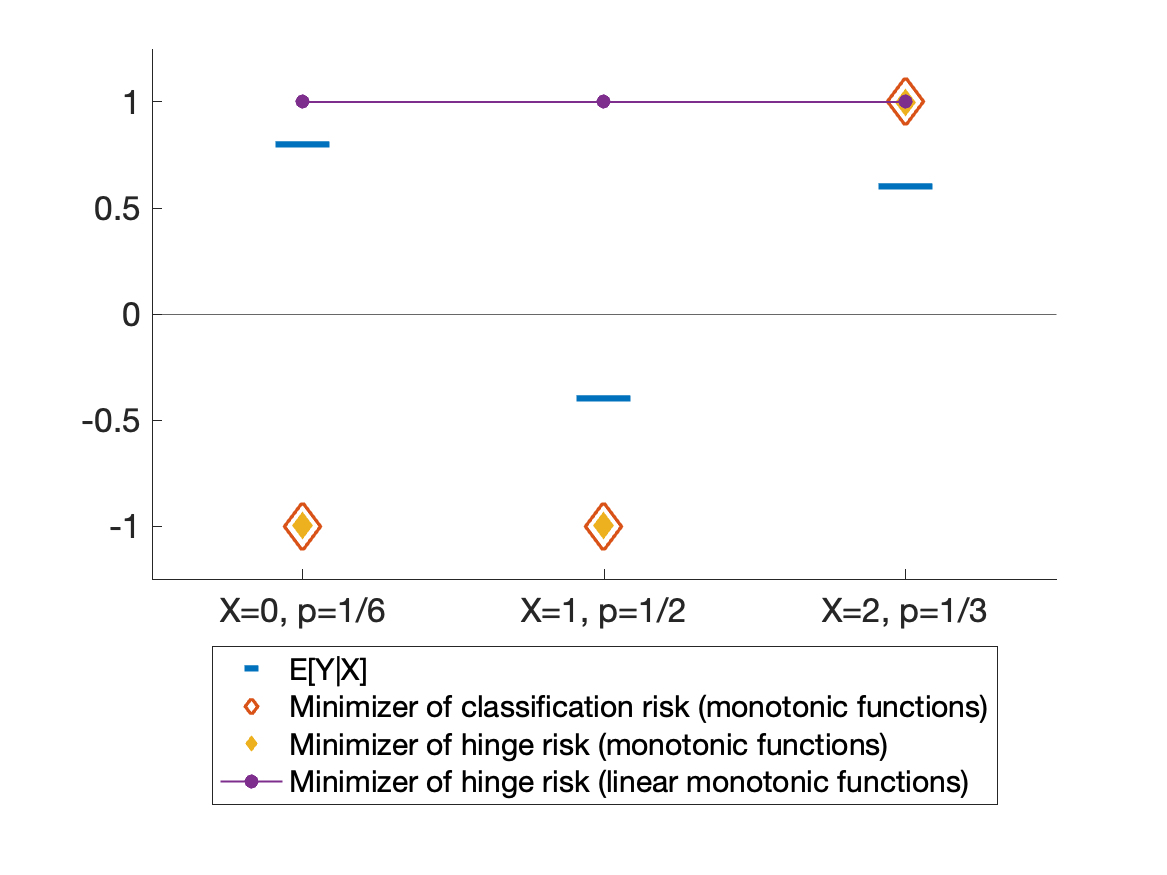}
\label{fig:monotonic example}
\end{tabular}
\end{center}
\end{figure}

\section{Constrained classification with surrogate loss}
\label{sec:constrained classification and surrogate loss approach}

Consider the binary classification problem of ascribing a binary label $Y \in \{-1,+1\}$ based on covariates $X \in \MX$, which are collectively distributed according to a joint distribution $P$. We let $X$ be a $d_x$-dimensional vector, $d_x< \infty$, and denote its marginal distribution by $P_{X}$. We denote the conditional probability of $Y=+1$ given $X=x$ by $\eta(x) \equiv P(Y = +1|X=x)$ and otherwise maintain the notation introduced in the Introduction. The ultimate objective is to minimize the classification risk of (\ref{eq:classification_risk}).

We study constrained classification problems that place restrictions on the allowed classifier functions, particularly on their prediction sets. Section \ref{sec:misspecification in constrained classification} introduces the surrogate loss approach for classification problems.   Section \ref{sec:MG-constrained classification} defines the $\MG$-constrained classification problem that places a constraint on the allowed prediction sets. This is a central problem throughout the paper.

\subsection{Misspecification in constrained classification and surrogate loss} \label{sec:misspecification in constrained classification}

Let $\MF$ be a constrained set of classifiers $f: \MX \rightarrow \Real$. If the set of classifiers were unconstrained, it is well known that the Bayes classifier defined by 
\begin{equation}
    f^{\ast}_{Bayes} = 2 \eta(x) - 1 \notag
\end{equation} 
minimizes the classification risk. Due to the constraints on the set of classifiers, however, the minimized classification risk on $\MF$ can be strictly larger than the first-best minimal risk $R(f^{\ast}_{Bayes})$. We refer to this situation as $R$-\textit{misspecification} of $\MF$: 

\bigskip{}

\begin{definition}[$R$-misspecification] \label{def:misspecification}
$\MF$ is $R$-misspecified if 
\begin{equation} \notag
\inf_{f \in \MF} R(f) > R(f^{\ast}_{Bayes}).
\end{equation}
If instead an equality holds, we say that $\MF$ is $R$-correctly specified.
\end{definition}

\bigskip{}

Because the 0-1 loss function is neither convex nor continuous, minimizing the empirical analogue of $R(f)$ is computationally challenging and often infeasible given the scale of the problems that we encounter in practice. Commonly used classification algorithms, such as boosting and support vector machines, replace the 0-1 loss with a surrogate loss function, $\phi:\Real\rightarrow\Real$, and aim to minimize the surrogate risk $R_{\phi}(f) \equiv E_P[\phi(Yf(X))]$.
Table \ref{tb:surrogate loss functions and their forms of H}  lists some commonly used surrogate loss functions including
the \textit{hinge loss} $\phi_h(\alpha)=\max\{0, 1-\alpha\}$, which corresponds to $\ell_1$-support vector machines, the \textit{exponential loss}  $\phi_e(\alpha)=\exp(-\alpha)$, which corresponds to AdaBoost, and the \textit{logistic loss}, which corresponds to neural network learning.

The seminal theoretical results that guarantee consistency of surrogate-risk classification \citep{Zhang_2004, Bartlett_et_al_2006, Nguyen_et_al_2009} crucially rely on the assumption that $\MF$ is $R$-correctly specified in the sense of Definition \ref{def:misspecification}. The central question that this paper poses is \textit{how is a surrogate loss approach affected if $\MF$ is $R$-misspecified?} This misspecification is a likely scenario, especially when the origins of the constraints have nothing to do with the assumptions on $P$, as is the case in the examples discussed in the next subsection. 

Since restricting the set of classifiers to those bounded between $-1$ and $1$ does not restrict their prediction sets, we can restrict attention to classifiers which take values $f(x) \in [-1,1]$ for all $x \in \MX$.
Following \citet{Bartlett_et_al_2006}, define the generic conditional $\phi$-risk function
for $a \in [-1,1]$, $\eta \in [0,1]$ as
\begin{equation*}
C_{\phi}(a, \eta) \equiv \eta \phi ( a ) +(1-\eta) \phi(-a ).    
\end{equation*}
At $x \in \MX$, the conditional $\phi$-risk of classifier $f$ is $C_{\phi} \left( f(x), \eta(x) \right)$.
It is also useful to define the optimal conditional $\phi$-risk attainable over all positive and negative predictions, and their difference,
\begin{equation}\label{def:DeltaC}
\begin{split}
C_{\phi}^{+}(\eta) & \equiv\inf_{a \in [0,1]} C_{\phi}(a,\eta),\\
C_{\phi}^{-}(\eta) & \equiv\inf_{a \in [-1,0)} C_{\phi}(a,\eta),\\
\Delta C_{\phi}(\eta) & \equiv C_{\phi}^{+}(\eta)-C_{\phi}^{-}(\eta).
\end{split}
\end{equation}
$\Delta C_{\phi}(\eta(x))$ captures the difference in conditional $\phi$-risk between the best positive and negative predictions $f(x)$.
In unconstrained problems, its sign determines whether a classifier minimizing conditional $\phi$-risk will have the same sign as the Bayes classifier.

Throughout the paper, we will limit our analysis to classification-calibrated loss functions \citep{Bartlett_et_al_2006}:

\bigskip{}

\begin{definition}[Classification-calibrated loss functions]
\footnote{This definition is equivalent to the one in  \citet{Bartlett_et_al_2006} except that we treat $f(x)=0$ as a positive prediction.}
\label{def:classification-calibrated loss functions}
Loss function $\phi$ is classification-calibrated if
\begin{align}
    \begin{split}
\Delta C_{\phi}(\eta) < 0 & \text{ for all } \eta \in (1/2,1], \\
\Delta C_{\phi}(\eta) > 0 & \text{ for all } \eta \in [0,1/2). \label{eq:classification_calibration}
\end{split} 
\end{align}
\end{definition}

\bigskip{}

Noting that the surrogate risk can be expressed as
\begin{equation}
    E_P[\phi(Yf(X))]=E_{P_X}[C_{\phi}(f(X), \eta(X))], \label{eq:surrogate risk expression}
\end{equation}
the definition of classification-calibrated loss functions implies that at every $x \in \MX$ with $\eta(x)\neq 1/2$, every $f(x)$ that minimizes $C_{\phi}(f(x), \eta(x))$ has the same sign as the Bayes classifier, $\sign(2\eta(x)-1)$.
\citet{Bartlett_et_al_2006} show that many commonly used surrogate loss functions, including those listed in Table \ref{tb:surrogate loss functions and their forms of H}, are classification-calibrated.\footnote{\citet{Bartlett_et_al_2006} also show that any convex loss function $\phi$ is classification-calibrated if and only if it is differentiable at $0$ and $\phi^{\prime}(0)<0$.}

\subsection{$\MG$-constrained classification}
\label{sec:MG-constrained classification}

Throughout the paper, we restrict the set of classifiers by requiring that their prediction sets belong to a prespecified nonempty class of measurable sets, $\MG\subseteq 2^{\MX}$. See motivating Examples \ref{exm:interpretable classification}--\ref{exm:fair classification} below.

We denote by
\[
\MF_{\MG}
\equiv
\left\{
  f:\MX\rightarrow[-1,1]: f\text{ is measurable and }G_f\in\MG
\right\}
\]
the set of classifiers whose prediction sets are constrained to $\MG$.
We refer to $\MF_{\MG}$ as the \textit{$\MG$-constrained set of classifiers}. In this definition, we restrict $f$ to be bounded and, without loss of generality, normalize its range to $[-1,1]$.
Apart from measurability, the class $\MF_{\MG}$ imposes no restriction on the functional form of $f$ other than the restriction on its prediction set and range.

The population target of the constrained classification problem is the infimum
\begin{align*}
    \inf_{f\in\MF_\MG}R(f).
\end{align*}
When this infimum is attained, a constrained optimal classifier is any
$f^{\ast}\in\arginf_{f\in\MF_\MG}R(f)$.  Proposition \ref{prop:MG-constrained misspecification relation} below is formulated in terms of minimizing sequences and therefore does not require attainment of either the classification-risk or surrogate-risk infimum.
We refer to the classification problem over $\MF_\MG$ as \textit{$\MG$-constrained classification}.

The specification of the class of prediction sets $\MG$ represents the fairness, interpretability, and other exogenous requirements that are desired for classification rules. Some examples follow.

\bigskip{}

\begin{example}[Interpretable classification]
\label{exm:interpretable classification}
Decision-makers may prefer simple decision or classification rules that are easily understood or explained even at the cost of harming prediction accuracy. This concept, often referred to as \textit{interpretable machine learning}, has been pursued, for instance, in the prediction analysis of recidivism \citep{Zeng_et_al_2017} and the decision on medical intervention protocol \citep{Zhang_et_al_2018}. An example is a linear classification rule, in which
$\MG$ is a class of half-spaces with linear boundaries in $\MX = \Real^{d_x}$, 
\begin{equation}
    \MG=\{x\in\Real^{d_x}:x^{T}\beta\geq0,\beta\in\Real^{d_x}\}. \notag
\end{equation}
Note that $\MF_{\MG}$ is not restricted to linear functions. Any nonlinear function $f$ is also included in $\MF_{\MG}$ if its prediction set $G_{f}$ is a half-space of $\MX$. A classification tree is another type of classification rule that is interpretable. See, e.g., \citet{Breiman84book}.
\end{example}

\bigskip{}

\begin{example}[Monotone classification]
\label{exm:monotone classification} The framework we study can accommodate monotonicity constraints on classification.
Formally, a monotonicity constraint corresponds to a partial order $\precsim$ on $\MX$, and any prediction set $G_{f}$ has to respect this partial order in the sense that if $x_{1} \precsim x_{2}$ and $x_{1} \in G_{f}$, then $x_{2}\in G_{f}$. Monotonicity constraints have been utilized in the classification of credit ratings \citep{Chen_Li_2014}, and in the assignment of job training in the context of policy learning \citep{MT17}. 
\end{example}

\bigskip{}

\begin{example}[Fair classification]
\label{exm:fair classification}  Specification of $\MG$ can accommodate some fairness constraints introduced
in the literature on \textit{fair classification}.
Let $A=\left\{ 0,1\right\} $ be an element of $X$ indicating a binary protected group variable (e.g., race, gender). The decision-maker wants to ensure fairness of classification by, for instance, equalizing the raw positive classification rate (known as statistical parity): $P_{X}\left(f(x)\geq0\mid A=1\right)=P_{X}\left(f(x)\geq0\mid A=0\right)$.
The classification problem embedding this constraint is equivalent to $\MG$-constrained classification with 
\begin{align*}
{\cal G} & =\left\{ G\in2^{{\cal X}}:P_{X}\left(X\in G\mid A=1\right)=P_{X}\left(X\in G\mid A=0\right)\right\} ,
\end{align*}
where $\MG$ depends on the conditional distribution $P_{X|A}$ in this case. This fairness constraint is studied by \citet{Calders_Verwer_2010},
\citet{Kamishima_et_al_2011}, \citet{Dwork_et_al_2012}, \citet{Feldman_et_al_2015},
among others.
Some other forms of
fairness constraint, such as equalized odds and equalized positive
predictive value as reviewed by \citet{Chouldechova_et_al_2018},
can be accommodated in our framework as well via an appropriate construction of $\MG$. 
\end{example}

\bigskip

In the $\MG$-constrained classification problem, $R$-correct specification of $\MF_{\MG}$ is necessary and sufficient for every surrogate-risk-minimizing sequence to be Bayes-risk consistent.

\bigskip{}

\begin{proposition}\label{prop:MG-constrained misspecification relation}
Let $\MG\subseteq 2^{\MX}$ be a nonempty class of measurable subsets of $\MX$. Suppose that $\phi:\Real\rightarrow\Real$ is Borel measurable, bounded from below on $[-1,1]$, and classification-calibrated. Then, for any distribution $P$ on $\{-1,1\}\times\MX$, the following statements are equivalent:
\begin{enumerate}[label=\textnormal{(\roman*)},leftmargin=2.5em]
  \item $\MF_{\MG}$ is $R$-correctly specified:
  \[
    \inf_{f\in\MF_{\MG}}R(f)=R(f_{Bayes}^{\ast}).
  \]

  \item For every sequence $\{f_m\}_{m\geq1}\subset\MF_{\MG}$,
  \[
    R_{\phi}(f_m)\longrightarrow
    \inf_{f\in\MF_{\MG}}R_{\phi}(f)
  \]
  implies
  \[
    R(f_m)\longrightarrow R(f_{Bayes}^{\ast}).
  \]
\end{enumerate}
\end{proposition}

\begin{proof}
See Appendix \ref{appx:proof 1}.
\end{proof}

\bigskip

Proposition \ref{prop:MG-constrained misspecification relation} shows that if $\phi$ is classification-calibrated, $f_{\phi}^{\ast}$ that minimizes the surrogate risk over $\MF_{\MG}$ leads to a globally optimal classifier in terms of the classification risk if and only if $\MF_{\MG}$ is $R$-correctly specified. 
This result justifies the approach of surrogate risk minimization when $\MF_\MG$ is a sufficiently rich set of classifiers so that $\MF_{\MG}$ can contain $f^{\ast}_{Bayes}$, since $R$-correct specification guarantees consistency of surrogate-risk-minimizing classifiers.  

However, if constraints defining $\MG$ are motivated by considerations that are independent of any beliefs about the underlying data-generating process (e.g., Examples \ref{exm:interpretable classification}--\ref{exm:fair classification} above), $R$-correct specification of $\MF_{\MG}$ is hard to justify. Therefore, an important question for our analysis is whether surrogate risk minimization can yield a classifier achieving $\inf_{f \in \MF_{\MG}} R(f)$ \textit{without} requiring $R$-correct specification of $\MF_{\MG}$. The following section gives an answer to this question.

\section{Calibration of $\MG$-constrained classification}
\label{sec:calibration of MG-constrained classification}

\subsection{Unweighted Classification}

This section investigates the risk consistency of surrogate risk minimization over $\MF_{\MG}$, where $\MF_{\MG}$ is now allowed to be $R$-misspecified. 
To begin our analysis, we recall the numerical example illustrated in Figure \ref{fig:G-constrained example} in Section \ref{sec:examples}, where $\MF_{\MG}$ with the monotonicity restriction has
$\MG  = \{\emptyset, \{2\}, \{1,2\},\{0,1,2\}\}$ and is $R$-misspecified. In this example, hinge loss is robust to $R$-misspecification of $\MF_{\MG}$, but exponential loss is not. To what extent can we generalize this finding? What conditions do we need to guarantee that surrogate-risk-minimizing classifiers are consistent with the second-best (constrained optimal) classification rule in terms of the classification risk? We answer these questions below. 

For any classifier $f$, we define the \textit{$\MG$-constrained excess risk} of $f$ as \begin{equation}
R(f)-\inf_{f\in\MF_{\MG}}R(f), \notag  
\end{equation}
which is the \textit{regret} of $f$ relative to the constrained optimum in terms of the classification risk.
Similarly, we define the $\MG$-constrained excess $\phi$-risk of $f$ as
\begin{align*}
    R_{\phi}(f) - \inf_{f \in \MF_{\MG}}R_{\phi}(f).
\end{align*}
Fix $G\in\MG$ and let 
\begin{align*}
\MF_{G}\equiv\{f:G_{f}=G,f(\cdot)\in[-1,1]\}
\end{align*}
be the set of classifiers that share the prediction set $G$. Then $\{\MF_{G}:G\in\MG\}$ forms a partition of $\MF_{\MG}$
indexed by the prediction sets and satisfies $\MF_{\MG}=\cup_{G\in\MG}\MF_{G}$. With this definition at hand, choosing a classifier from $\MF_{\MG}$ can
be decomposed into two steps: choosing a prediction set $G$ from
$\MG$ and then choosing a classifier $f$ from $\MF_{G}$.

Denote the classification risk evaluated at prediction set $G$ by $\MR(G)\equiv\inf_{f\in\MF_{G}}R(f)$. 
Any $f\in\MF_{G}$ attains the same classification risk, so $\MR(G)=R(f)$ holds for all $f\in\MF_{G}$. $\MR(G)$ can
be written as 
\begin{align}
\MR(G) & =\int_{\MX}\left[\eta(x)1\{x\notin G\}+(1-\eta(x))1\{x\in G\}\right]dP_{X}(x),\nonumber \\
 & =\int_{\MX}(1-2\eta(x))\cdot1\{x\in G\}dP_{X}(x)+P(Y=1).\label{eq:classification risk at G}
\end{align}
Similarly, we define the surrogate risk evaluated at $G$ as $\MR_{\phi}(G)\equiv\inf_{f\in\MF_{G}}R_{\phi}(f)$, which can be written as 
\begin{align}
\MR_{\phi}(G) & =\inf_{f\in\MF_{G}}\int_{\MX}\left[\eta(x)\phi(f(x))+(1-\eta(x))\phi(-f(x))\right]dP_{X}(x)\nonumber \\
 & =\int_{G}\inf_{f(x)\in [0,1]}C_{\phi}(f(x),\eta(x))dP_{X}(x)+\int_{G^{c}}\inf_{f(x) \in [-1, 0) }C_{\phi}(f(x),\eta(x))dP_{X}(x),\nonumber
\end{align}
where the second line follows from the fact that $f \in \MF_{G}$ is unconstrained other than via its prediction set and that minimization over $f\in\MF_{G}$ can be performed pointwise at each $x$. 
For $f\in\MF_{G}$, $f(x)$ is constrained to $[0,1]$ for $x\in G$
and to $[-1,0)$ for $x\in G^{c}$. 
Using definitions (\ref{def:DeltaC}), the surrogate risk at $G$ can be written as 
\begin{align}
\MR_{\phi}(G) & =\int_{\MX}\left[C_{\phi}^{+}(\eta(x))\cdot1\{x\in G\}+C_{\phi}^{-}(\eta(x))\cdot1\{x\notin G\}\right]dP_{X}(x)\nonumber \\
 & =\int_{\MX}\Delta C_{\phi}(\eta(x))\cdot1\{x\in G\}dP_{X}(x)+\int_{\MX}C_{\phi}^{-}(\eta(x))dP_{X}(x).\label{eq:simplified surrogate risk}
\end{align}
By comparing the expressions of the risks in (\ref{eq:classification risk at G}) and (\ref{eq:simplified surrogate risk}), we obtain the first main theorem that establishes a sufficient condition for the surrogate risk $\MR_{\phi}(G)$ to reproduce the same ordering of the risk over $G\in\MG$ as the classification risk $\MR(G)$.

\bigskip{}

\begin{theorem}[Global calibration of $\MG$-constrained excess risk] \label{thm:risk equivalence} 
Let $P$ be an arbitrary distribution on $\{-1,1\}\times \MX$ and $\MG \subseteq 2^{\MX}$ be a class of measurable subsets of $\MX$. For $G,G'\in\MG$, the risk ordering $\MR(G)\geq\MR(G')$ in terms of classification risk
is equivalent to 
\begin{equation}
\int_{G\setminus G'}(1-2\eta(x))dP_{X}(x)\geq\int_{G'\setminus G}(1-2\eta(x))dP_{X}(x),\label{eq:thm eq1}
\end{equation}
while the risk ordering $\MR_{\phi}(G)\geq\MR_{\phi}(G')$ in terms of surrogate risk
is equivalent to 
\begin{equation}
\int_{G\setminus G'}\Delta C_{\phi}(\eta(x))dP_{X}(x)\geq\int_{G'\setminus G}\Delta C_{\phi}(\eta(x))dP_{X}(x).\label{eq:thm eq2}
\end{equation}
Hence, if $\Delta C_{\phi}(\eta(x))$ is proportional to $1-2\eta(x)$ up
to a positive constant, i.e., 
\begin{align}
\Delta C_{\phi}(\eta(x))=c(1-2\eta(x))\mbox{ for some }c>0,\label{eq:conditon_risk equivalence}
\end{align}
the risk ordering over $\MG$ by surrogate risk
$\MR_{\phi}(G)$ agrees with the risk ordering over $\MG$
by classification risk $\MR(G)$ for any $P$.

In particular, when $\phi$ is the hinge loss $\phi_h(\alpha) =  \max \{0, 1-\alpha \}$,
\begin{equation}
\Delta C_{\phi}(\eta(x))=1-2\eta(x)\label{eq:H difference for hinge} \notag
\end{equation}
holds, establishing that hinge risk preserves the risk ordering of  classification risk.
\end{theorem} 
\begin{proof}
By equation (\ref{eq:classification risk at G}), 
\begin{align*}
 \MR(G)-\MR(G')
= & \int_{\MX}(1-2\eta(x))\cdot[1\{x\in G\}-1\{x\in G'\}]dP_{X}(x)\\
= & \int_{\MX}(1-2\eta(x))\cdot[1\{x\in G\setminus G'\}-1\{x\in G'\setminus G\}]dP_{X}(x)\\
= & \int_{G\setminus G'}(1-2\eta(x))dP_{X}(x)-\int_{G'\setminus G}(1-2\eta(x))dP_{X}(x).
\end{align*}
This proves (\ref{eq:thm eq1}), the first claim of the theorem.

Given the representation of the surrogate risk shown in (\ref{eq:simplified surrogate risk}),
a similar argument yields (\ref{eq:thm eq2}), the second claim of
the theorem.

For the hinge loss $\phi_h(\alpha)=\max\{0,1-\alpha\}$ and $f\in\MF_{G}$,
we have 
\begin{align*}
C_{\phi_h}(f(x),\eta(x))=(1-2\eta(x))f(x)+1.
\end{align*}
Hence, we obtain 
\begin{align}
C_{\phi_h}^{+}(\eta) & =\begin{cases}
2-2\eta & \text{for \ensuremath{\eta>1/2},}\\
1 & \text{for \ensuremath{\eta\leq1/2},}
\end{cases}\nonumber \\
C_{\phi_h}^{-}(\eta) & =\begin{cases}
1 & \text{for \ensuremath{\eta>1/2},}\\
2\eta & \text{for \ensuremath{\eta\leq1/2}.}\notag
\end{cases}
\end{align}
Hence, $\Delta C_{\phi_h}(\eta)=1-2\eta$ holds for all $\eta\in[0,1]$. 
\end{proof}
\bigskip{}

Theorem \ref{thm:risk equivalence} does not exploit the condition that $\phi$ is classification-calibrated, but if a surrogate loss function satisfies condition (\ref{eq:conditon_risk equivalence}), it is automatically classification-calibrated.

\bigskip{}

\begin{remark} Many commonly used surrogate loss functions do
not satisfy condition (\ref{eq:conditon_risk equivalence}) in
Theorem \ref{thm:risk equivalence}. Table \ref{tb:surrogate loss functions and their forms of H}
shows the forms of $\Delta C_{\phi}(\eta)$ for hinge loss, exponential loss, logistic
loss, quadratic loss, and truncated quadratic loss functions. With the exception of the hinge loss function, none of these functions satisfy condition (\ref{eq:conditon_risk equivalence}).
That is, among the surrogate-loss-based algorithms that are commonly used in practice, the $\ell_1$-support vector machine corresponding to hinge loss is the only algorithm whose surrogate risk preserves the classification risk.

\begin{table}[h]
\centering \caption{Surrogate loss functions and their associated forms for $\Delta C_{\phi}$}
\label{tb:surrogate loss functions and their forms of H} 
\begin{threeparttable}
\scalebox{0.9}{ %
\begin{tabular}{c:c:c}
\hline 
Loss function  & $\phi(\alpha)$  & $\Delta C_{\phi}\left(\eta\right)$\tabularnewline
\hline 
0-1 loss  & $1\{\alpha\leq0\}$  & $1-2\eta$\tabularnewline
\hdashline 
Hinge loss  & $ \max\{0,1-\alpha\}$  & $1-2\eta$\tabularnewline
\hdashline 
Exponential loss  & $e^{-\alpha}$  & $\begin{cases}
\begin{array}{c}
-2\sqrt{\eta(1-\eta)}+1\\
2\sqrt{\eta(1-\eta)}-1
\end{array} & \begin{array}{l}
\mbox{if }0\leq\eta<1/2\\
\mbox{if }1/2\leq\eta \leq 1
\end{array}\end{cases}$\tabularnewline
\hdashline 
Logistic loss  & $\log(1+e^{-\alpha})$  & $\begin{cases}
\begin{array}{c}
\log(2\eta^{\eta}(1-\eta)^{1-\eta})\\
-\log(2\eta^{\eta}(1-\eta)^{1-\eta})\\
\end{array} & \begin{array}{l}
\mbox{if }0\leq\eta < 1/2\\
\mbox{if }1/2\leq\eta\leq 1\\
\end{array}\end{cases}$\tabularnewline
\hdashline 
Quadratic loss  & $(1-\alpha)^{2}$  & $\begin{cases}
\begin{array}{c}
(1-2\eta)^{2}\\
-(1-2\eta)^{2}
\end{array} & \begin{array}{l}
\mbox{if }0\leq\eta <1/2\\
\mbox{if }1/2\leq\eta\leq 1
\end{array}\end{cases}$\tabularnewline
\hdashline 
Truncated quadratic loss  & $(\max\{0,1-\alpha\})^{2}$  & $\begin{cases}
\begin{array}{c}
(1-2\eta)^{2}\\
-(1-2\eta)^{2}
\end{array} & \begin{array}{l}
\mbox{if }0\leq\eta <1/2\\
\mbox{if }1/2\leq\eta\leq 1
\end{array}\end{cases}$\tabularnewline
\hline 
\end{tabular}}
\end{threeparttable}
\end{table}

\end{remark}

\bigskip{}

The well known inequality by \citet{Zhang_2004} relates the excess surrogate risk to the excess classification risk under $R$-correct specification. 
As a corollary of Theorem \ref{thm:risk equivalence}, if we set $\phi = \phi_h$, we can generalize Zhang's inequality by allowing $R$-misspecification of the classifiers. 

\bigskip{}

\begin{corollary}\label{corr:zhang's inequality} 
For any distribution $P$ on $\{-1,1\}\times \MX$ and a class of measurable subsets $\MG \subseteq 2^{\MX}$, if $\Delta C_{\phi}(\eta(x))=c(1-2\eta(x))$ for some $c>0$, then the
following inequality
\begin{align*}
c \left( R(f)-\inf_{f\in\MF_{\MG}}R(f) \right) \leq R_{\phi}(f)-\inf_{f\in\MF_{\MG}}R_{\phi}(f)
\end{align*}
holds for any $f\in\MF_{\MG}$. 
\end{corollary} 
\begin{proof}
Since $R(f) = \MR(G_f)$ and 
$\inf\limits_{f\in\MF_{\MG}}R(f) = \inf\limits_{G\in\MG} \MR(G)$,
\begin{align*}
c \left( R(f)-\inf_{f\in\MF_{\MG}}R(f) \right) = c \left( \MR(G_f)-\inf_{G\in\MG}\MR(G) \right).
\end{align*}
Take any $\eps>0$, then there exists $G^* \in \MG$ s.t.
$\MR(G^*) < \inf\limits_{G\in\MG} \MR(G) + c^{-1}\eps$, hence
\begin{equation*}
c \left( \MR(G_f)-\inf_{G\in\MG}\MR(G) \right) < c \left( \MR(G_f)-\MR(G^*) \right) + \eps.
\end{equation*}
Using equations (\ref{eq:classification risk at G}), (\ref{eq:simplified surrogate risk}), and the assumption that $\Delta C_{\phi}(\eta(x))=c(1-2\eta(x))$ for some $c>0$, we have that
\begin{align*}
c\left( \MR(G_f)-\MR(G^*) \right)
& = 
\int_{G_f\setminus G^*}c(1-2\eta(x))dP_{X}(x) -
\int_{G^*\setminus G_f}c(1-2\eta(x))dP_{X}(x) \\
& =
\int_{G_f\setminus G^*} \Delta C_{\phi}(\eta(x)) dP_{X}(x) -
\int_{G^*\setminus G_f} \Delta C_{\phi}(\eta(x)) dP_{X}(x) \\
& =
\MR_{\phi}(G_f) - \MR_{\phi}(G^*).
\end{align*}
Combining with 
$\MR_{\phi}(G_f) \leq R_{\phi}(f)$
and
$\MR_{\phi}(G^*) \geq \inf\limits_{G \in \MG} \MR_{\phi}(G) = \inf\limits_{f\in\MF_{\MG}}R_{\phi}(f)$
we obtain
\begin{align*}
c \left( R(f)-\inf_{f\in\MF_{\MG}}R(f) \right)
& =
c \left( \MR(G_f)-\inf_{G\in\MG}\MR(G) \right) \\
& <
c \left( \MR(G_f)-\MR(G^*) \right) + \eps \\
& =
\MR_{\phi}(G_f) - \MR_{\phi}(G^*) + \eps \\
& \leq
R_{\phi}(f) - \inf_{f\in\MF_{\MG}}R_{\phi}(f) + \eps.
\end{align*}
Since 
$c \left( R(f)-\inf\limits_{f\in\MF_{\MG}}R(f) \right) < R_{\phi}(f) - \inf\limits_{f\in\MF_{\MG}}R_{\phi}(f) + \eps$
for any $\eps>0$, the claim of the corollary follows.
\end{proof}
\bigskip{}

Corollary \ref{corr:zhang's inequality} shows that if surrogate
loss $\phi$ satisfies condition (\ref{eq:conditon_risk equivalence}),
then the classifier $f_{\phi}^{\ast}$ that minimizes the surrogate risk over $\MF_{\MG}$ also minimizes the classification risk over $\MF_{\MG}$. 
Importantly, this result holds without assuming the $R$-correct specification of $\MF_{\MG}$. It justifies the use of hinge loss in the constrained classification problem irrespective of whether or not $\MF_{\MG}$  is correctly $R$-specified. Note, however, that the result relies on the fact that at every $x \in \MX$ we can choose any $f(x) \in [-1,1]$ as long as the prediction set constraint $G_{f} \in \MG$ is satisfied. We relax this requirement in the next section.

Further analysis can show that the condition (\ref{eq:conditon_risk equivalence})
in Theorem \ref{thm:risk equivalence} is not only sufficient but
also necessary. To formally show this, we adopt the concept of universal equivalence of loss functions introduced by \citet{Nguyen_et_al_2009} to the current setting.

\bigskip{}

\begin{definition}[Universal equivalence] 
Loss functions $\phi_{1}$ and $\phi_{2}$ are universally equivalent,
denoted by $\phi_{1}\overset{u}{\sim}\phi_{2}$,
if for any $\MX$ and any distribution $P$ on $\MY\times\MX$
\begin{equation*}
{\MR}_{\phi_{1}}\left(G_{1}\right)\leq{\MR}_{\phi_{1}}\left(G_{2}\right)  \Leftrightarrow
{\MR}_{\phi_{2}}\left(G_{1}\right)\leq{\MR}_{\phi_{2}}\left(G_{2}\right)
\end{equation*}
holds for any $G_{1},G_{2}\in 2^{\MX}$.
\end{definition}

\bigskip{}

Universally equivalent loss functions $\phi_{1}$ and $\phi_{2}$ lead to the same risk ordering over $\MG$. Hence, if a loss function $\phi$ is universally equivalent to the 0-1 loss, the $\phi$-risk shares the same risk ordering with the classification risk.

The following theorem establishes a necessary and sufficient condition for two classification-calibrated loss functions to be universally equivalent.

\bigskip{}

\begin{theorem}
\label{thm:univesal equivalence}
Let $\phi_{1}$ and $\phi_{2}$ be classification-calibrated loss functions.
Then $\phi_{1}\overset{u}{\sim}\phi_{2}$ if and only if for some constant $c>0$,
$\Delta C_{\phi_{2}}\left(\eta\right)=c\Delta C_{\phi_{1}}\left(\eta\right)$
for all $\eta\in\left[0,1\right]$.
\end{theorem}

\begin{proof}
See Appendix \ref{appx:proof 1}.
\end{proof}

\bigskip{}

The `if' part of the theorem is a generalization of Theorem \ref{thm:risk equivalence} in that it does not assume that either of $\phi_1$ or $\phi_2$ is the 0-1 loss function.

When we set $\phi_{2}$ to the 0-1 loss function, Theorem \ref{thm:univesal equivalence} yields the class of loss functions that are universally equivalent to the 0-1 loss functions. This class exactly coincides with the class of loss functions that satisfy the condition (\ref{eq:conditon_risk equivalence}) in Theorem
\ref{thm:risk equivalence}. Hence, the following corollary holds.

\bigskip{}

\begin{corollary} \label{corr:universally equivalence} A classification-calibrated
loss function $\phi$ is universally equivalent to the 0-1 loss function
if and only if $\phi$ satisfies condition (\ref{eq:conditon_risk equivalence})
for any $\eta(x)\in[0,1]$. In particular, the hinge loss
$\phi_h(\alpha)=\max\{0,1-\alpha\}$ and its positive affine transformations
are universally equivalent to the 0--1 loss function.
\end{corollary} 

\bigskip{}

We conclude this section with a remark to compare our constrained classification framework to that of \citet{Nguyen_et_al_2009}.

\bigskip{}

\begin{remark} \label{remark:Nguyen_et_al_2009} \citet{Nguyen_et_al_2009} show that, for the classification problem in which an optimal pair comprising a quantizer and a classifier is to be chosen,
the hinge loss function is also the only surrogate loss function that preserves the consistency of surrogate loss classification. In their framework, the quantizer is a stochastic mapping $Q\in\mathcal{Q}:{\cal X}\to {\cal Z}$,
where ${\cal Z}$ is a discrete space and ${\cal Q}$ is a possibly
constrained class of conditional distributions of $Z$ given $X$,
$Q\left(Z\mid X\right)$. The classifier is a function $\gamma\in\Gamma:{\cal Z}\to \mathbb{R}$,
where $\Gamma$ is the set of all measurable functions on ${\cal Z}$. The motivation for using $Z$ as an input, instead of $X$, is to reduce the dimension of $X$, which might be a high-dimensional vector. \citet{Nguyen_et_al_2009} propose estimating the pair $\left(Q,\gamma\right)\in\mathcal{Q}\times\Gamma$ that minimizes the risk $R\left(\gamma,Q\right){\equiv}P\left(Y\neq\mbox{sign}\left(\gamma\left(Z\right)\right)\right)$, by solving the surrogate loss classification problem: $\inf_{(Q,\gamma)\in\mathcal{Q}\times\Gamma}R_{\phi}(Q,\gamma)$, where $R_{\phi}(Q,\gamma)=E\phi(Y\gamma(Z))$. They show that, among the commonly used surrogate loss functions, only hinge loss classification leads to the optimal pair of $(Q,\gamma)$.

The framework we study is different from that of \citet{Nguyen_et_al_2009}, and neither nests the other. The framework \citet{Nguyen_et_al_2009}
study constrains the mapping $Q:{\cal X}\to{\cal Z}$, whereas the framework we study constrains prediction sets $G_{f}$ for all classifiers $f$. Furthermore, the set of classifiers $\Gamma$ considered in \citet{Nguyen_et_al_2009} contains the Bayes classifier, whereas the set of classifiers $\MF_{\MG}$ we consider may not contain the Bayes classifier. 
\end{remark}

\section{Consistency of hinge risk classification with functional form constraints}
\label{sec:classification preserving reduction}

The previous section considers $\MF_{\MG}$, the class of all functions whose prediction sets are in $\MG$. The generalized Zhang's inequality shown in Corollary \ref{corr:zhang's inequality} heavily relies on the richness of $\MF_{\MG}$. This richness, however, limits the computational
attractiveness of a surrogate-loss approach, since convexity in optimization of an empirical analogue of the surrogate risk does not directly follow from $\MF_{\MG}$, and typically requires additional functional form restrictions for $f$. 

Unfortunately, once a functional form restriction on $f$ is imposed on top of the prediction set constraint $G_f \in \MG$, the global calibration property of the hinge risk shown in Theorem \ref{thm:risk equivalence} may break down. This is demonstrated in the numerical example shown in Figure \ref{fig:monotonic example}, where hinge risk minimization with linear monotone functions (purple line) fails to achieve the second-best optimality.
This example illustrates that even with hinge loss, consistency to the second best classifier becomes a fragile property once the functional form of $f$ is constrained in addition to the prediction set constraint $G_{f}\in\MG$. 
Consequently, it is natural to ask what additional functional form restriction we can safely introduce to $\MF_{\MG}$ without threatening consistency, i.e., for which subclass $\widetilde{\MF}_{\MG} \subset \MF_{\MG}$ does minimizing the hinge risk $R_h(f)$ over $f \in \widetilde{\MF}_{\MG}$ lead to a classifier that minimizes the classification risk $R(f)$ over $f \in \MF_{\MG}$?

Formally, we introduce the following definition of \textit{classification-preserving reduction of $\MF_{\MG}$}. 

\bigskip{}

\begin{definition}[Classification-preserving reduction]
\label{def:Constrained classification preserving reduction}
A subset of classifiers $\widetilde{\MF}_{\MG} \subseteq \MF_{\MG}$ is a 
\emph{classification-preserving reduction} of $\MF_{\MG}$ if, for any 
distribution $P$ on $\{-1,1\}\times\MX$ such that $\arg\min_{G\in\MG} R(G)\neq\varnothing$, a minimizer of $R_h(f)$ over 
$\widetilde{\MF}_{\MG}$ exists and any $\tilde{f}^{\ast} \in 
    \argmin_{f \in \widetilde{\MF}_{\MG}} R_h(f)$
satisfies
\[
    R(\tilde{f}^{\ast})
    =
    \inf_{f \in \MF_{\MG}} R(f).
\]
\end{definition}

\bigskip{}

If the subclass $\widetilde{\MF}_{\MG}$ is a classification-preserving reduction of $\MF_\MG$, the classifier that minimizes the hinge risk over $\widetilde{\MF}_{\MG}$ also minimizes the classification risk over $\MF_\MG$. Hence, if restricting attention to $\widetilde{\MF}_{\MG}$ simplifies optimization, this justifies using the smaller class $\widetilde{\MF}_{\MG}$
 instead of the richer class $\MF_{\MG}$. Note that Definition \ref{def:Constrained classification preserving reduction} requires the classification-preserving property only for distributions under which the constrained classification problem attains its optimum, i.e., $\arg\min_{G\in\MG}R(G)\neq\varnothing$. In contrast, the analysis in Section~\ref{sec:constrained classification and surrogate loss approach} does not require attainment of the constrained classification-risk infimum.

For any $G \in \MG$, we define $\bar{f}_{G}(x) \equiv 2 \cdot 1\{x \in G\} -1$, a step function over $\MX$ that indicates $x \in G$ and $x \notin G$ with values $+1$ and $-1$, respectively.
For any classifier $f$, we define its upper level sets at $t \in [-1,1]$ as
\begin{equation}
G_{f} (t) \equiv \left\{x : f(x) \geq t\right\}.
\end{equation}
The following theorem is the second main theorem of the paper that establishes a simple-to-check sufficient condition for a reduced set of classifiers $\widetilde{\MF}_{\MG}$ to be a classification-preserving reduction of $\MF_\MG$.
\bigskip{}

\begin{theorem}[Consistency under classification-preserving reduction] \label{thm:continuous functions} 
Given a class of measurable subsets $\MG \subseteq 2^{\MX}$ and $\MF_{\MG}=\left\{ f:G_{f}\in\MG,\ f(\cdot)\in[-1,1]\right\}$, suppose $\widetilde{\MF}_{\MG} \subset \MF_{\MG}$ satisfies the following two conditions:  
\begin{enumerate}[label=(A\arabic*)]
    \item  \label{asm:level set condition}
    For every $f\in\widetilde{\MF}_{\MG}$, $G_{f} (t)\in\MG$ for all $t\in(-1,1]$;
    \item \label{asm:optimizer condition}
    For any $G \in \MG$, $\bar{f}_{G} \in\widetilde{\MF}_{\MG}$.
\end{enumerate}
Then the following claims hold:\\
(i) For any distribution $P$ on $\{-1,1\}\times\MX$ such that
$\arg\min_{G\in\MG}R(G)\neq\varnothing$, every $G^\ast\in\arg\min_{G\in\MG}R(G)$
satisfies
$\bar f_{G^\ast}\in
\argmin_{f\in\widetilde{\MF}_{\MG}}R_h(f)$
and
$\min_{f\in\widetilde{\MF}_{\MG}}R_h(f)
=
2\min_{G\in\MG}R(G)$.
\\
(ii)
$\widetilde{\MF}_{\MG}$ is a classification-preserving reduction of $\MF_{\MG}$.
\end{theorem}

\begin{proof}
See Appendix \ref{appx:proof 1}.
\end{proof}
\bigskip{}

The theorem establishes that the two conditions \ref{asm:level set condition} and \ref{asm:optimizer condition} are sufficient for  $\widetilde{\MF}_{\MG}$ to be a classification-preserving reduction of $\MF_{\MG}$. 
The conditions \ref{asm:level set condition} and \ref{asm:optimizer condition} in Theorem \ref{thm:continuous functions}  are simple to interpret and guarantee the consistency of the hinge risk minimization, but they do not imply a general way to construct a classification-preserving reduction that makes the empirical hinge risk minimization a convex program. 
For monotone classification analyzed in Section \ref{sec:Applications to monotone classification}, we propose a way to construct $\widetilde{\MF}_{\MG_M}$ which is a classification-preserving reduction and a convex set of classifiers, reducing the empirical hinge risk minimization to a convex program. 

\bigskip{}

\begin{remark}[Monotone
classification with a class of monotone functions]
\label{example:Monotonic classification with a class of monotonic functions} Hinge risk
minimization embedding a monotonicity restriction remains consistent when we use a class of monotone functions. Let $\precsim$ be a partial order on $\MX$, and let $\MG_{M}$ be the collection of all $G\in2^{\MX}$ that respect monotonicity (i.e., if $x_{1}\precsim x_{2}$ and $x_{1}\in G$, then $x_{2}\in G$). 
Define $\widetilde{\MF}_{\MG_{M}}$ as a class of functions $f:\MX\rightarrow[-1,1]$ that are weakly monotone with respect
to $\precsim$ (i.e., satisfying $f(x_{1})\leq f(x_{2})$ if $x_{1}\precsim x_{2}$). Then the prediction set of any $f\in\widetilde{\MF}_{\MG_{M}}$ respects the partial order $\precsim$ (i.e., if $x_{1}\precsim x_{2}$ and $x_{1}\in G_{f}$, then $x_{2}\in G_{f}$). 
For any $t\in[-1,1]$ and $f\in\widetilde{\MF}_{\MG_{M}}$, $G_{f}(t)=\{x:x\precsim\tilde{x}\mbox{ for any }\tilde{x}\mbox{ such that }f(\tilde{x})=t\}\in\MG_{M}$ holds,
satisfying condition \ref{asm:level set condition} in Theorem \ref{thm:continuous functions}. 
In addition, for any $G\in\MG_{M}$, since $\bar{f}_{G}(x)=2\cdot1\{x\in G\} -1$
is weakly monotone in $\precsim$, $\bar{f}_{G} \in \widetilde{\MF}_{\MG_{M}}$
holds, satisfying condition \ref{asm:optimizer condition} in Theorem \ref{thm:continuous functions}. Hence $\widetilde{\MF}_{\MG_{M}}$
is a classification-preserving reduction of $\MF_{\MG_{M}}$. Therefore, according to Theorem \ref{thm:continuous functions},
hinge risk minimization over $\widetilde{\MF}_{\MG_{M}}$ yields an optimal classifier in terms of the classification risk. Section \ref{sec:Applications to monotone classification} focuses on monotone classification and investigates its statistical and computational properties. 
\end{remark}

\bigskip{}

\begin{remark}[Necessity of an exact hinge risk minimizer in a classification-preserving reduction]
\label{remark:Necessity of an exact hinge risk minimizer in a classification-preserving reduction}
Theorem \ref{thm:continuous functions} relies on obtaining an exact minimizer of surrogate risk over $\widetilde{\MF}_\MG$,
Finding a sequence of near-minimizers of surrogate risk does not guarantee convergence
to the minimum classification risk. There is no result analogous to Corollary \ref{corr:zhang's inequality} or
Zhang's inequality for a classification-preserving reduction.
We demonstrate it with a simple counterexample.

Take $\MX=\{1,2\}$ and $\MG=\{\emptyset,\{2\},\{1,2\}\}$.
$\MG$ imposes monotonicity of the prediction sets, and we consider
the set of classifiers $\widetilde{\MF}_\MG$ with $f(x) \in [-1,1]$ and $f(1) \leq f(2)$,
shown in Remark \ref{example:Monotonic classification with a class of monotonic functions} to be a classification-preserving reduction.
Take uniform $P_X$: $P_X(1) = P_X(2) = 1/2$ and
let $\eta(1) = P(Y=1\mid X=1)=1$, $\eta(2)=P(Y=1\mid X=2)=0$.
The classification risk of set $G$ equals
$\MR(G) = 1/2 \cdot 1\{1 \notin G \} + 1/2 \cdot 1\{2 \in G \}$,
thus $\MR(\emptyset) = \MR(\{1,2\}) = 1/2 = \MR^{\ast}$ and $\MR(\{2\}) = 1$.

Consider now the hinge loss of a classifier $f \in \widetilde{\MF}_\MG$ under $P$:
$$
R_h (f) = E_P(1-Yf(X)) = \frac{1}{2}(1-f(1)) + \frac{1}{2}(1+f(2))
= 1 + \frac{f(2)-f(1)}{2}.
$$
The minimum hinge risk over $\widetilde{\MF}_\MG$ equals $2\MR^{\ast}=1$ and is achieved by
all classifiers of the form $f(1)=f(2) = t, t \in [-1,1]$. All such classifiers
also minimize classification risk, having prediction sets $\emptyset$
or $\{1,2\}$.

However, this result is razor-thin. Consider a sequence of classifiers $f_n$ in $\tilde{\MF}_{\MG}$ of the form
$f_n (1)= -1/n$, $f_n(2) = 1/n$.
These are small perturbation of the hinge risk minimizer $f(1) = f(2) = 0$ as $n \to \infty$.
They have prediction sets $G_{f_n} = \{2\}$
with classification risk $R(f_n) = 1$.
Their hinge risks equal $R_h(f_n) = 1 + 1/n$.
Thus, as $n \to \infty$,
this sequence of classifiers has excess hinge risk converging to zero,
but their excess classification risk remains far from zero.

This counterexample shows that if we specify $\tilde{\MF}_{\MG}$ as an open set, $\{ f(\cdot) \in [-1,1]: f(1) < f(2) \}$ and, accordingly. $\argmin_{f \in \widetilde{\MF}_\MG} R_h(f)$ is empty, the conclusion of Theorem \ref{thm:continuous functions} fails. The requirement in Definition \ref{def:Constrained classification preserving reduction} that $\argmin_{f \in \widetilde{\MF}_\MG} R_h(f)$
exists is thus necessary and not just a simplification.
\end{remark}

\section{Extensions to weighted classification and individualized treatment rules}
\label{sec:Extension to individualized treatment rules}

This section extends the primary results obtained in Sections \ref{sec:calibration of MG-constrained classification} and \ref{sec:classification preserving reduction} to the weighted classification introduced in Section \ref{sec:Connection and contributions to causal policy learning} and causal policy learning. 
We refer to $R^{\omega}$ and $R_{\phi}^{\omega}$ defined in (\ref{eq:weighted_classification_risk}) and (\ref{eq:weighted_surrogate_risk}) as weighted classification risk and weighted $\phi$-risk, respectively.
In the weighted classification problem, we view $P$ as a joint distribution of $(\omega,Y,X) \in\Real_{+} \times \{-1,1\} \times \MX$.

Assume that the marginal distribution of $X$, $P_X$, has a
probability density $p_X$ with respect to some dominating measure $\mu(X)$, and that $\omega > 0$ with probability greater than 0.\footnote{When $\omega = 0$ almost surely, the weighted classification risk $R^{\omega}(f)$ and the weighted $\phi$-risk $R_{\phi}^{\omega}(f)$ are degenerate to zero for any $f$, making the comparison of risks trivial. } For a non-negative weight variable $\omega$, let 
\begin{align*}
    \omega_{+}(x) \equiv E_{P}\left[\omega\mid Y=+1, X=x\right],\ \ 
    \omega_{-}(x) \equiv E_{P}\left[\omega\mid Y=-1, X=x\right].
\end{align*}

We can show that for arbitrary $P$ with $E_P[\omega]>0$, we can transform $P$ to a probability measure $\tilde{P}$ of $(Y,X) \in \{-1, 1 \} \times \MX$ such that the weighted $\phi$ classification risk $R_{\phi}^{\omega}(\cdot)$ under $P$ can be calibrated by the unweighted $\phi$-risk $R_{\phi}(\cdot)$ under $\tilde{P}$ up to a positive constant. This equivalence holds for any surrogate losses $\phi$ including the 0-1 classification loss. Such $\tilde{P}$ can be constructed as
\begin{align*}
    \tilde{p}_{X}(x) \equiv \frac{E_{P}[\omega | x]}{E_P[\omega
    ]}p_{X}(x),\ \tilde{\eta}(x) = \widetilde{P}(Y=+1|x) \equiv \begin{cases}
\begin{array}{c}
\frac{\eta(x)\omega_{+}(x)}{E_{P}[\omega|x]}\\
\eta(x)
\end{array} & \begin{array}{c}
\mbox{if\ } E_{P}[\omega|x] >0\\
\mbox{otherwise}
\end{array}\end{cases}.
\end{align*}
To explicitly show the dependence of $R_{\phi}^{\omega}$ and $R_{\phi}$ on the underlying distribution, we denote $ R_{\phi}^{\omega}(P,f) = E_{P}[\omega\cdot \phi(Yf(X))]$ and $R_{\phi}(P,f) = E_{P}[\phi(Yf(X))]$ for any distribution $P$. Then weighted $\phi$-risk of classifier $f$ under distribution $P$ satisfies
\begin{align}
    R_{\phi}^{\omega}(P,f) &= E_{P}[\omega\cdot \phi(Yf(X))] \nonumber\\
    &= \int_{\MX} [\eta(x)\omega_{+}(x)\phi(f(x)) + (1 - \eta(x))\omega_{-}(x)\phi(-f(x))]p_{X}(x) d\mu(x) \nonumber \\
         &= \int_{\MX} [E_{P}[\omega|x]\tilde{\eta}(x)\phi(f(x)) + E_{P}[\omega|x](1 - \tilde{\eta}(x))\phi(-f(x))]p_{X}(x) d\mu(x) \nonumber\\
     &= \int_{\MX} [\tilde{\eta}(x)\phi(f(x)) + (1 - \tilde{\eta}(x))\phi(-f(x))]E_{P}[\omega]\tilde{p}_{X}(x) d\mu(x) \nonumber \\
    &= E_{P}[\omega]E_{\widetilde{P}}[\phi(Yf(X))]\nonumber \\
    &= E_{P}[\omega] R_\phi\left(\widetilde{P},f\right), \label{eq:unweighted_representation}
\end{align}
where the third equality follows from $(1-\eta(x))\omega_{-}(x) = E_{P}[\omega|x](1-\tilde{\eta}(x))$, which is derived from $E_{P}[\omega|x] = \eta(x)\omega_{+}(x) + (1-\eta(x))\omega_{-}(x)$.

Consequently, the weighted $\phi$-risk of prediction set $G$ in the weighted classification
problem under $P$ is proportional to the $\phi$-risk in the unweighted classification problem
under $\widetilde{P}$:
\begin{align*}
     \MR_{\phi}^{\omega}(P,G) \equiv \inf_{f \in \MF_G} R_{\phi}^{\omega}(P,f) = E_{P}[\omega] \inf_{f \in \MF_G} R_{\phi}(\widetilde{P},f) = E_{P}[\omega] \MR_{\phi}(\widetilde{P},G).
\end{align*}
Therefore the $\phi$-risk ranking of prediction sets in the weighted classification problem
under $P$ is the same as in the unweighted problem under $\widetilde{P}$ for any $\phi$. 

Thus, we can extend our main results -- Corollary \ref{corr:zhang's inequality}, Theorem \ref{thm:univesal equivalence}, and Theorem \ref{thm:continuous functions} -- to the weighted classification problems by applying those results to the unweighted classification risk under $\widetilde{P}$.

\medskip

\begin{corollary}\label{cor:zhang's inequality_WC}
For any distribution $P$ on $\Real_{+}\times \{-1,1\} \times \MX$
satisfying $0<E_P[\omega]<\infty$, and any surrogate loss function
$\phi$ for which there exists a constant $c>0$ such that
$
\Delta C_{\phi}(u)=c(1-2u)
$
for every $u\in[0,1]$,
\begin{align}
c\left(
R^{\omega}(P,f)
-\inf_{g\in\MF_{\MG}}R^{\omega}(P,g)
\right)
\leq
R_{\phi}^{\omega}(P,f)
-\inf_{g\in\MF_{\MG}}R_{\phi}^{\omega}(P,g)
\label{eq:zhangs inequality_WC}
\end{align}
holds for any $f\in\MF_{\MG}$.
\end{corollary}

\begin{proof}
    The result follows by combining Corollary \ref{corr:zhang's inequality} and the result (\ref{eq:unweighted_representation})
\end{proof}

\medskip

\begin{corollary}\label{corr:risk equivalence_WC}
Let $\phi_{1}$ and $\phi_{2}$ be classification-calibrated loss functions.
Then
\begin{align*}
\MR^{\omega}_{\phi_{1}}(P,G_{1})
\leq
\MR^{\omega}_{\phi_{1}}(P,G_{2})
\quad\Longleftrightarrow\quad
\MR^{\omega}_{\phi_{2}}(P,G_{1})
\leq
\MR^{\omega}_{\phi_{2}}(P,G_{2})
\end{align*}
holds for every distribution $P$ on
$\Real_{+}\times\{-1,1\}\times\MX$ satisfying
$0<E_P[\omega]<\infty$, every class of measurable subsets
$\MG\subseteq 2^{\MX}$, and every $G_1,G_2\in\MG$
if and only if there exists a constant $c>0$ such that
$
\Delta C_{\phi_{2}}(u)
=
c\,\Delta C_{\phi_{1}}(u)
$
for all $u\in[0,1]$.
\end{corollary}

\begin{proof}
    The result follows by combining Theorem \ref{thm:univesal equivalence} and the result (\ref{eq:unweighted_representation})
\end{proof}

\medskip

\begin{corollary}\label{thm:continuous functions_WC} Given a distribution $P$ on $\Real_{+}\times \{-1,1\} \times \MX$ and a class of measurable subsets $\MG \subseteq 2^{\MX}$, suppose that $\widetilde{\MF}_{\MG} \subset \MF_{\MG}$ satisfies the conditions \ref{asm:level set condition} and \ref{asm:optimizer condition} in Theorem \ref{thm:continuous functions}. 
Let $\tilde{f}^{\ast} \in \arginf_{f \in \widetilde{\MF}_{\MG}} R_{\phi_h}^{\omega}(P,f)$. Then the following claims hold:\\
(i) The classifier $\tilde{f}^{\ast}$  minimizes the weighted classification risk $R^{\omega}(P,\cdot)$ over $\MF_{\MG}$. \\
(ii) Let $\MR^{\omega}(P,G) \equiv \inf_{f \in \MF_G} R^{\omega}(P,f)$. For any $G^{*}\in \arginf_{G \in \MG} R^{\omega}(P,f)$, the function $\bar{f}_{G^{*}}$ is a minimizer of $R_{\phi_h}^{\omega}(P,\cdot)$ over $\widetilde{\MF}_{\MG}$. 
\end{corollary}

\begin{proof}
    The results follow by combining Theorem \ref{thm:continuous functions} and the result (\ref{eq:unweighted_representation})
\end{proof}

\section{Statistical properties}
\label{sec:statistical property}

The analyses presented so far concern the consistency of a surrogate loss approach in terms of the population risk criterion. 
It is important to note that Theorems \ref{thm:risk equivalence} and \ref{thm:continuous functions} do not impose any restriction on the underlying distribution of $(Y,X)$. Accordingly, equivalence of the risk orderings and risk-minimizing classifiers between the classification and hinge risks remains valid even if we consider empirical analogues of the risks constructed from the empirical distribution of the sample. 
Theorems \ref{thm:risk equivalence} and \ref{thm:continuous functions} hence guarantee that a classifier minimizing the empirical hinge risk over $\MF_{\MG}$ or over a classification-preserving reduction $\widetilde{\MF}_{\MG}$ also minimizes the empirical classification risk. This discussion also applies to Corollaries \ref{corr:risk equivalence_WC} and \ref{thm:continuous functions_WC} for the weighted classification.

In this section, we assess the generalization performance of hinge-risk-minimizing classifiers, allowing for general misspecification of the constrained set of classifiers. Toward that goal, let $\MG$ be fixed, and consider a set of classifiers $\widetilde{\MF}_{\MG}$ that satisfies the conditions of Theorem \ref{thm:continuous functions}.
We consider the general case of weighted classification.
Let $\{\left(\omega_{i},Y_{i},X_{i}\right) \in\Real_{+} \times \{-1,1\} \times \MX :i=1,\ldots,n\}$
be a sample of observations that are independent and identically
distributed (i.i.d.) as $(\omega,Y,X)$. We denote the joint distribution
of a sample of $n$ observations by $P^{n}$ and the expectation with respect
to $P^{n}$ by $E_{P^{n}}[\cdot]$. 

We define the empirical weighted classification risk
and empirical weighted hinge risk, respectively, as 
\begin{align*}
\widehat{R}^{\omega}(f)&\equiv  \frac{1}{n}\sum_{i=1}^{n}\omega_{i}1\left\{ Y_{i} \cdot \text{sign} (f\left(X_{i}\right) ) \leq0\right\} , \nonumber\\
\widehat{R}_{\phi_{h}}^{\omega}(f)&\equiv \frac{1}{n}\sum_{i=1}^{n}\omega_{i}\max\{0,1-Y_{i}f\left(X_{i}\right)\} = \frac{1}{n}\sum_{i=1}^{n}\omega_{i}\left(1-Y_{i}f\left(X_{i}\right) \right), 
\end{align*}
where the max operator in the hinge loss is redundant if we constrain $f(\cdot)$ to $[-1, 1]$.
Let $\hat{f}$ be a classifier that minimizes $\widehat{R}_{\phi_{h}}^{\omega}(\cdot)$
over $\widetilde{\MF}_{\MG}$. We evaluate the statistical properties
of $\hat{f}$ in terms of the excess weighted classification risk relative to the minimal risk over $\MF_{\MG}$, $R^{\omega}(\hat{f})-\inf_{f\in\MF_{\MG}}R^{\omega}(f)$. In particular, we later derive a distribution-free upper bound on the mean of the excess weighted classification risk.

Since $\widetilde{\MF}_{\MG}$ is a classification-preserving reduction of $\MF_{\MG}$, Corollary \ref{thm:continuous functions_WC} with $P$ being the empirical distribution of the sample $\{\left(\omega_{i},Y_{i},X_{i}\right):i=1,\ldots,n\}$ shows that the minimizer of the empirical hinge risk, $\hat{f} \in \arginf_{f \in \widetilde{\MF}_{\MG}}\widehat{R}_{\phi_h}^{\omega}(f)$, also minimizes the empirical risk $\widehat{R}^{\omega}(\cdot)$ over $\MF_{\MG}$. Hence we can apply existing results of statistical learning theory \citep[e.g.,][]{Mammen_Tsybakov_1999,Tsybakov_2004,MT17} to show statistical properties of the excess weighted classification risk $R^{\omega}(\hat{f})-\inf_{f\in\MF_{\MG}}R^{\omega}(f)$.

Let $H_{1}^{B}\left(\epsilon,\MG,P_{X}\right)$ be the $L_{1}\left(P_X\right)$-bracketing entropy of a class of prediction sets $\MG$, with a detailed definition provided in Definition \ref{def:bracketing entropy} in Appendix \ref{appx:proof 2}.\footnote{With a slight abuse of notation, we denote by $H_{1}^{B}\left(\epsilon,\MG,P_{X}\right)$ the bracketing entropy number of the class of indicator functions, $H_{1}^{B}\left(\epsilon,\MH_\MG,P_{X}\right)$, where $\MH_{\MG}\equiv \{1\{\cdot \in G\}: G \in \MG\}$.} 
The following theorem gives a distribution-free upper bound on  the mean of the $\MG$-constrained excess weighted classification risk in terms of the bracketing entropy.   

\bigskip

\begin{theorem}\label{thm:statistical propery for excess classification risk}
Let $\widetilde{\MF}_\MG$ be a subclass of $\MF_{\MG}$ that satisfies conditions \ref{asm:level set condition} and \ref{asm:optimizer condition} in Theorem \ref{thm:continuous functions}. 
Suppose that $\MP$ is a class of distributions on $\Real_{+} \times \{-1,1\}\times \MX$ such that, uniformly over $P \in \MP$, (i) $|\omega| < M$ a.s. for some positive constant $M$, (ii) there exists some constant $A>0$,
where the distribution of $X$ is continuous with density bounded
above by $A$, and (iii) there exist positive constants $C$ and $r$ for which
\begin{align}
H_{1}^{B}\left(\epsilon,\MG,P_{X}\right) \leq C\epsilon^{-r} \label{eq:bracketing entoropy conditoin_G}
\end{align}
holds for all $0<\epsilon\leq1$.
Then, for $\hat{f} \in \arginf_{f \in \widetilde{\MF}_{\MG}}\widehat{R}_{\phi_h}^{\omega}(f)$, the following holds:
\begin{align}
\sup_{P\in\mathcal{P}}E_{P^{n}}\left[R^{\omega}(\hat{f})-\inf_{f\in\MF_{\MG}}R^{\omega}(f)\right] = O\left(\tau(n)\right), \label{eq:upper bound_statistical property_admissible refinement}
\end{align}
where $\tau(n)= n^{-1/2}$ if $r < 1$, $\tau(n) = n^{-1/2}\log(n)$ if $r=1$, and $\tau(n) = n^{-1/(1+r)}$ if $r >1$.
\end{theorem}

\begin{proof}
See Supplemental Appendix.
\end{proof}

\medskip

Theorem~\ref{thm:statistical propery for excess classification risk} shows that the mean of the $\MG$-constrained excess weighted classification risk converges to zero at a rate determined by the exponent $r$ in the bracketing entropy condition
\eqref{eq:bracketing entoropy conditoin_G}. \cite{Dudley1999} exhibits many examples that satisfy this bracketing entropy conditions. In the next section, we apply the theorem to monotone
classification.

\section{Applications to monotone classification}
\label{sec:Applications to monotone classification}

This section specializes the general results in Sections~\ref{sec:calibration of MG-constrained classification}--%
\ref{sec:statistical property} to the monotone classification problem introduced in Example~\ref{exm:monotone classification}.
In view of Theorem~\ref{thm:risk equivalence}, we focus on hinge loss throughout this section. We assume that $\MX$ is a compact subset of $\Real^{d_x}$, where $d_x<\infty$.

We equip $\MX$ with the coordinatewise partial order. For
$x=(x_1,\ldots,x_{d_x})^\top$ and $\tilde{x}=(\tilde{x}_1,\ldots,\tilde{x}_{d_x})^\top$, we write
$x\leq\tilde{x}$ if $x_j\leq\tilde{x}_j$ for every $j=1,\ldots,d_x$. We further write $x<\tilde{x}$ if $x\leq\tilde{x}$ and
$x\neq\tilde{x}$.

We consider the class of monotone prediction sets
\[
\MG_M
\equiv
\left\{
G\subseteq\MX:
x\in G \text{ and } x\leq\tilde{x}
\implies \tilde{x}\in G
\text{ for all }x,\tilde{x}\in\MX
\right\}.
\]
Thus, every $G\in\MG_M$ is an upper set with respect to the
coordinatewise partial order. Accordingly, the class of bounded
non-decreasing classifiers is defined as
\[
\MF_M
\equiv
\left\{
f:\MX\to[-1,1]:
f(x)\leq f(\tilde{x})
\text{ whenever }x\leq\tilde{x}
\right\}.
\]

We consider weighted hinge risk minimization given $\MF_{M}$. Note that $\MF_M$ is a classification-preserving reduction of $\MF_{\MG_M}$ (see Remark \ref{example:Monotonic classification with a class of monotonic functions}). Let $\hat{f}_{M}$ be a minimizer of $\widehat{R}_{\phi_h}^{\omega}(\cdot)$ over $\MF_{M}$. Since the weighted hinge risk for classifiers constrained to $-1 \leq f(x) \leq 1$ gives the weighted linear loss $\omega \phi_h(yf(x)) = \omega\cdot (1-yf(x))$, minimization of the empirical weighted hinge risk over $\MF_M$ can be formulated as the following linear programming:
\begin{align}
&\max_{f_{1},\ldots,f_{n}}	\sum_{i=1}^{n}\omega_{i}Y_{i}f_{i} \label{eq:LP_monotone classification}\\
\mbox{s.t.} \quad	&f_{i}\geq f_{j} \mbox{ whenever }X_{i}\geq X_{j}; \notag \\
\quad	&-1\leq f_{i}\leq1\mbox{ for }1\leq i\leq n, \notag    
\end{align}
where the first inequality constraints correspond to the monotononicity constraint on $\MF_M$, and the second inequality constraints correspond to the range constraint for $f \in \MF_M$. Solving this linear program yields the values of $\hat{f}_{M}$ at the values of $x$ observed in the training sample.  
Let $\left(\hat{f}_{M}\left(X_{1}\right),\ldots,\hat{f}_{M}\left(X_{n}\right)\right)$ be the solution of (\ref{eq:LP_monotone classification}). 
Then any function in $\MF_{M}$ that passes the points $\left(\left(X_{1},\hat{f}_{M}\left(X_{1}\right)\right),\ldots,\left(X_{n},\hat{f}_{M}\left(X_{n}\right)\right)\right)$ minimizes the empirical weighted hinge risk over $\MF_M$.\footnote{All classifiers obtained from this procedure predict a unique label at each point $x$ observed in the training sample, whereas they may not give a unique prediction at a point $x$ not observed in the training sample. One possible way to predict a label at an unobserved point $x$ without violating the monotonicity constraint is to predict its label by the largest label among those predicted by all classifiers in $\arginf_{f \in \MF_M}\widehat{R}_{\phi_h}(f)$. 
Let $\widetilde{\MX}$ be a set of $x$ observed in the training sample.
Fixing $\hat{f}_M \in \arginf_{f \in \MF_M}\widehat{R}_{\phi_h}(f)$, this way of predicting a label is equivalent to predicting the label of $x \in \MX \backslash \widetilde{\MX}$ as the sign of $\min\{\hat{f}_{M}(\tilde{x}):\tilde{x}\in \widetilde{\MX}, \tilde{x} \geq x\}$ if there exists $\tilde{x}\in \widetilde{\MX}$ such that $\tilde{x} \geq x$, and as 1 otherwise.}
Since $\MF_{M}$ is a classification-preserving reduction of $\MF_{\MG_M}$, Theorem \ref{thm:continuous functions} with $P$ replaced by the empirical distribution of the sample shows that any solution to (\ref{eq:LP_monotone classification}) exactly minimizes $\widehat{R}_{\phi_h}^{\omega}(\cdot)$ over $\MF_{\MG_M}$.

We next investigate the statistical properties of this procedure.
Since $\MF_M$ is a classification-preserving reduction of $\MF_{\MG_M}$,  we can apply Theorem \ref{thm:statistical propery for excess classification risk}. 
Toward this goal, we first characterize an upper bound on the bracketing entropy number of the class of monotone prediction sets $\MG_M$. The next lemma, which we borrow from Theorems 8.3.1--2 in \cite{Dudley1999}, gives an upper bound on the $L_{1}\left(P_X\right)$-bracketing entropy of $\MG_{M}$. 
Here, we assume that $X$ is continuously distributed with bounded density.

\bigskip{}

\begin{lemma} \label{lemma:bracketing entropy_monotone G}
Suppose that $P_{X}$ is absolutely continuous with respect to the Lebesgue measure on $\MX$ and has a density that is bounded from above by a finite constant $A>0$. Then there exist constants $C$ and $C^\prime$, which depend only on $A$, such that 
\begin{align*}
H_{1}^{B}\left(\epsilon,\MG_{M},P_{X}\right) \leq \begin{cases}
\begin{array}{c}
C\epsilon^{1-d_{x}}\\
C^{\prime}\epsilon^{-1/2}
\end{array} & \begin{array}{c}
(d_{x}\geq2)\\
(d_{x}=1)
\end{array}\end{cases}.
\end{align*}
holds for all $0<\epsilon\leq1$.
\end{lemma}

\begin{proof}
See Supplemental Appendix.
\end{proof}

\bigskip{}

Applying Lemma~\ref{lemma:bracketing entropy_monotone G} with $r=d_x-1$ when $d_x\geq2$ and with $r=1/2$ when
$d_x=1$, Theorem~\ref{thm:statistical propery for excess classification risk} yields a uniform upper bound on the mean of the
$\MG_M$-constrained excess weighted classification risk of $\hat f_M$. The following theorem shows that the excess weighted
classification risk of the estimator obtained from the linear program in~\eqref{eq:LP_monotone classification} converges at the same rate as the welfare regret of monotone treatment rules established by \citet{MT17}.

\bigskip{}

\begin{theorem} \label{thm:nonparametric monotone classification}
Let $\mathcal{P}$ be a class of distributions on 
$\Real_{+}\times \MY\times\MX$ such that, uniformly over $P\in\mathcal{P}$,
the marginal distribution $P_X$ is absolutely continuous with respect to
the Lebesgue measure on $\MX$ and has a density bounded from above by
some finite constant $A>0$, and
$0\leq \omega \leq M$ $P\text{-a.s.}$
for some finite constant $M>0$. Then, for
$\hat{f}_{M} \in \arginf_{f\in \MF_M}
\widehat{R}_{\phi_{h}}^{\omega}(f)$,
\begin{align*}
\sup_{P\in\mathcal{P}}
E_{P^{n}}\left[
R^{\omega}(\hat{f}_{M})
-\inf_{f\in\MF_{\MG_{M}}}R^{\omega}(f)
\right]
= O\left(\tau(n)\right),
\end{align*}
where $\tau(n)=n^{-1/2}$ if $d_x=1$,
$\tau(n)=n^{-1/2}\log(n)$ if $d_x=2$, and
$\tau(n)=n^{-1/d_x}$ if $d_x\geq3$.
\end{theorem}

\begin{proof}
Since $\MF_M$ satisfies conditions \ref{asm:level set condition} and \ref{asm:optimizer condition} in Theorem \ref{thm:continuous functions} with $\MG$ equal to $\MG_M$ (Remark \ref{example:Monotonic classification with a class of monotonic functions}), the result follows from Theorem \ref{thm:statistical propery for excess classification risk} and Lemma \ref{lemma:bracketing entropy_monotone G}.
\end{proof}

\bigskip{}

This theorem guarantees the consistency of monotone classification using hinge loss and the class of monotone classifiers $\MF_{M}$.  The rate of convergence corresponds to $\tau(n)$.

\section{Empirical illustration}\label{sec:empirical application}

We illustrate our hinge risk minimization approach using experimental
data from the National Job Training Partnership Act (JTPA) Study. The study
randomly assigned applicants to eligibility for a program that provides job training, job-search
assistance, and other services for 18 months. It collected background
information on the applicants prior to random assignment, as well as data
on their earnings for 30 months following assignment. See
\citet{Bloom1997} for a detailed description of the study.

We revisit the setup of \citet{KT18}. The analysis sample contains
9,223 observations with complete information on treatment assignment,
earnings, education, and pre-program earnings. The probability of assignment
to program eligibility is $2/3$. Following the cost-adjusted specification in
\citet{KT18}, our outcome $Z$ is total earnings over the 30 months after
assignment, less \$774 for individuals assigned to treatment. This amount is
the average cost of services per treatment assignment used in their
application and accounts for incomplete take-up of program services. 
We base treatment allocation on years of education and earnings in the year prior to assignment.

As in \citet{MT17}, we consider treatment allocations that prioritize
individuals with less education and lower pre-program earnings. That is, letting
$x=(x_1,x_2)$ be education and pre-program earnings, respectively, we consider the following set of monotone classifiers 
\[
    \MF_M
    =
    \left\{
    f:\MX\to[-1,1]:
    f(x)\geq f(\tilde{x})
    \text{ whenever }x\leq\tilde{x}
    \right\},
\]
where the inequality $x \leq \tilde{x}$ is element-wise.\footnote{We can still apply our theoretical results after reversing the coordinatewise order.} 
\citet{MT17} approximate the class of monotone allocations by sieves of monotone classifiers with progressively more flexible piecewise-linear monotone boundaries, whereas our method directly optimizes over monotone assignments using linear programming, without requiring the choice of a sieve dimension.

To evaluate the welfare performance of the learned policy out of sample, we randomly split the data into a training sample and an evaluation sample in an approximately $3:1$ ratio. 
The training sample consists of 6,917 observations and is used to estimate the monotone policy $\widehat f_M$. The evaluation sample consists of 2,306 observations and is used to evaluate the welfare of the estimated policy and other benchmark policies. 
We estimate $\widehat f_M$ by weighted hinge risk minimization using the linear program in \eqref{eq:LP_monotone classification} with the monotonicity constraints reversed.

Figure~\ref{fig:treatment rule_exc} displays the treatment
allocation estimated with the training sample. 
The rule assigns treatment to 77.71\% of individuals in the evaluation sample.\footnote{At covariate values not observed in the training sample, an individual is assigned to treatment if both education and pre-program earnings are weakly lower than those of at least one treated observation in the training sample.}

We evaluate the welfare gain relative to the no-treatment rule $f_{NT}$ and the uniform assignment rule $f_T$ as a benchmark. 
Both welfare gains and the standard errors are estimated with the evaluation sample. Reported confidence intervals are based on the normal approximation.

Table~\ref{tab:welfare_gain_jtpa} reports the welfare performance of the learned monotone rule and the uniform assignment rule. The estimated monotone rule yields an estimated welfare gain of \$625.7 per population member.
Despite the monotonicity constraint, the estimated welfare gain of the monotone assignment rule is much larger than \$67.3, the estimated welfare gain of the uniform assignment rule, although
the 95\% confidence interval for each welfare gain includes zero.

\bigskip

\begin{figure}[htbp]
\centering
\caption{Monotone treatment allocation learned from the JTPA data}
\label{fig:treatment rule_exc}

\includegraphics[width=0.80\textwidth]{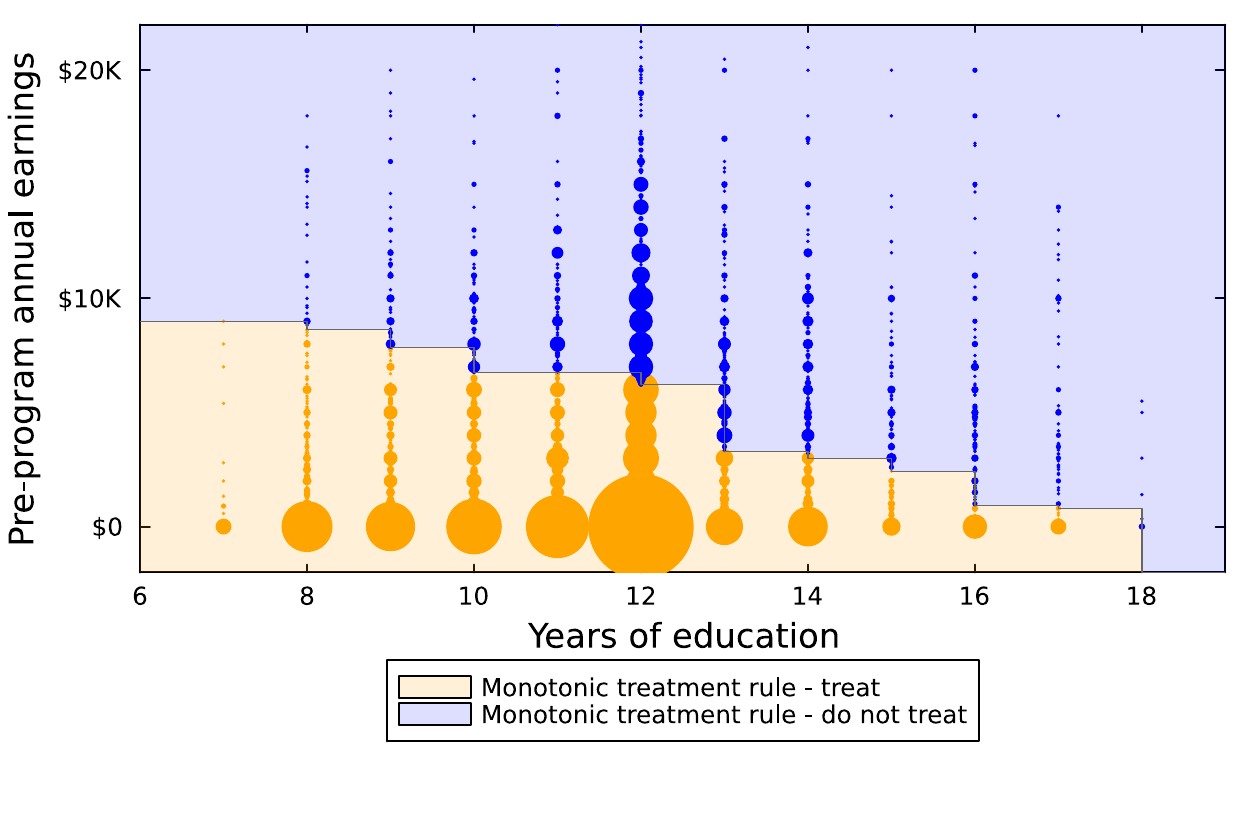}

\begin{minipage}{0.85\textwidth}
\footnotesize
Notes: The figure displays the monotone treatment rule estimated with
the training sample. The horizontal axis shows years of education, and
the vertical axis shows pre-program annual earnings. The orange-shaded
area represents the treatment region, while the blue-shaded area
represents the no-treatment region. Each circle represents a distinct
covariate value in the learning sample, with its size proportional to
the number of observations at that value.
\end{minipage}

\end{figure}

\begin{table}[htbp]
\centering
\caption{Out-of-sample welfare gains of treatment assignment rules}
\label{tab:welfare_gain_jtpa}
\begin{tabular}{lcc}
\hline
Treatment rule
    & Share treated
    & Estimated welfare gain \\
    & 
    & per population member \\
\hline
Treat everyone
    & 100\%
    & \$67 \\
    &
    & ($-\$1,425$, \$1,559) \\

Monotone rule $\widehat f_M$
    & 77.7\%
    & \$626 \\
    &
    & ($-\$618$, \$1,869) \\

\hline
\end{tabular}

\begin{minipage}{0.85\textwidth}
\footnotesize
Notes: Welfare gains are measured in dollars per population member
relative to the no treatment policy and are net of a \$774 treatment cost.
Parentheses report two-sided 95\% confidence intervals.
The monotone rule is estimated with the training sample and evaluated with the evaluation sample.
\end{minipage}
\end{table}

\section{Conclusion}

This paper studies the consistency of surrogate risk minimization approaches to classification and weighted classification when the set of classifiers is constrained, where weighted classification subsumes policy learning for individualized treatment assignment rules. Our focus is on how surrogate-risk-minimizing classifiers behave if the constrained set of classifiers violates the assumption of correct specification. Our first main result shows that when the constraint restricts classifiers' prediction sets only, hinge losses are the only loss functions that secure consistency of the surrogate-risk-minimizing classifier without the assumption of correct specification. When the constraint additionally restricts the functional form of the classifiers, the surrogate-risk-minimizing classifier is not generally consistent even with hinge loss. Our second main result shows that, in this case, conditions \ref{asm:level set condition} and \ref{asm:optimizer condition} in Theorem \ref{thm:continuous functions} become sufficient for the consistency of the hinge-risk-minimizing classifier. 

This paper also investigates the statistical properties of hinge-risk-minimizing classifiers in terms of uniform upper bounds on the excess regret. Exploiting hinge loss and the class of monotone classifiers in the monotone classification problem, we show that the empirical surrogate-risk-minimizing classifier can be computed using linear programming. All of the results obtained in the standard classification setting extend naturally to the weighted classification problem, so that our contributions carry over to causal policy learning and related applications.

\newpage

\appendix
\part*{Appendix}

\section{Proofs of the results in Sections \ref{sec:constrained classification and surrogate loss approach}--\ref{sec:classification preserving reduction}}
\label{appx:proof 1}

This appendix provides proofs of the results in Sections
\ref{sec:constrained classification and surrogate loss approach}--\ref{sec:classification preserving reduction} alongside some auxiliary lemmas. We begin with the auxiliary facts needed for Proposition~\ref{prop:MG-constrained misspecification relation}.

\begin{lemma}
\label{lem:conditional-surrogate-risk-sequential}
Suppose that $\phi:\Real\rightarrow\Real$ is bounded from below on $[-1,1]$ and classification-calibrated. Define
$
  q_{\phi}(u):=|\Delta C_{\phi}(u)|
$ for  $u\in[0,1]$.
Then the following statements hold:
\begin{enumerate}[label=\textnormal{(\roman*)},leftmargin=2.5em]
  \item The functions $C_{\phi}^{+}$ and $C_{\phi}^{-}$ are finite and continuous on $[0,1]$. Consequently, $\Delta C_{\phi}$ and $q_{\phi}$ are continuous on $[0,1]$.

  \item We have
    $\Delta C_{\phi}(1/2)=0$, $q_{\phi}(1/2)=0$, and $q_{\phi}(u)>0$ for $u\neq1/2$.

  \item The value
    $M_{\phi}:=\max_{u\in[0,1]}q_{\phi}(u)$
  is finite and strictly positive. Moreover, for every $\varepsilon>0$, there exists $a\in(0,1]$ such that 
  \[
    |1-2u|\leq a
    \quad\Longrightarrow\quad
    q_{\phi}(u)\leq\varepsilon.
  \]

  \item For every $\delta\in(0,1]$, define
  $
    \kappa_{\phi}(\delta)
    :=
    \inf_{\substack{u\in[0,1]:|1-2u|\geq\delta}}
    q_{\phi}(u).
  $
  Then the infimum is attained and
  \begin{equation}
  \label{eq:kappa-positive-sequential}
    \kappa_{\phi}(\delta)
    =
    \min_{\substack{u\in[0,1]:|1-2u|\geq\delta}}
    q_{\phi}(u)
    >0.
  \end{equation}
\end{enumerate}
\end{lemma}

\begin{proof}
    See Supplemental Appendix.
\end{proof}

\medskip

For the remainder of the proof, maintain the assumptions of Proposition
\ref{prop:MG-constrained misspecification relation}, and let
\[
  \MF_{\mathrm{all}}
  :=
  \{f:\MX\rightarrow[-1,1]\text{ measurable}\},
  \qquad
  G_B:=\{x\in\MX:\eta(x)\geq1/2\}.
\]
Write
$
  R^{\ast}:=R(f_{Bayes}^{\ast})
  =\inf_{f\in\MF_{\mathrm{all}}}R(f)
$ and 
$
  R_{\phi}^{\ast}
  :=\inf_{f\in\MF_{\mathrm{all}}}R_{\phi}(f).
$
For a measurable set $G\subseteq\MX$, define
\[
  \MF_G:=\{f\in\MF_{\mathrm{all}}:G_f=G\},
  \qquad
  \mathcal R(G):=\inf_{f\in\MF_G}R(f),
  \qquad
  \mathcal R_{\phi}(G):=\inf_{f\in\MF_G}R_{\phi}(f).
\]
Classification risk depends only on the prediction set. Direct comparison
with $G_B$ therefore gives
\begin{equation}
\label{eq:sequential-classification-gap}
  \mathcal R(G)-R^{\ast}
  =
  \int_{G\triangle G_B}|1-2\eta(x)|dP_X(x).
\end{equation}

We next justify the analogous surrogate-risk representation without assuming
that any conditional infimum is attained. Let $A_+=[0,1]$,
$A_-=[-1,0)$, and
\[
  V_s(u):=\inf_{a\in A_s}C_{\phi}(a,u),
  \qquad s\in\{+,-\}.
\]
Thus $V_+=C_{\phi}^{+}$ and $V_-=C_{\phi}^{-}$. Fix $\rho>0$. For every
$u\in[0,1]$ and $s\in\{+,-\}$, choose $a_{s,u}\in A_s$ such that
$
  C_{\phi}(a_{s,u},u)\leq V_s(u)+\rho/2.
$

By Lemma \ref{lem:conditional-surrogate-risk-sequential}(i), both
$v\mapsto C_{\phi}(a_{s,u},v)$ and $V_s(v)$ are continuous. Hence the same
action remains within $\rho$ of $V_s(v)$ on a neighborhood of $u$.
Compactness of $[0,1]$ gives a finite subcover. Assigning each point to the
first neighborhood in this subcover yields a Borel-measurable,
piecewise-constant map $\alpha_{s,\rho}:[0,1]\to A_s$ satisfying
\begin{equation}
\label{eq:sequential-measurable-near-minimizer}
  C_{\phi}(\alpha_{s,\rho}(u),u)
  \leq V_s(u)+\rho
  \qquad(u\in[0,1]).
\end{equation}
For a measurable $G$, define
\[
  f_{G,\rho}(x)
  :=
  \alpha_{+,\rho}(\eta(x))1\{x\in G\}
  +\alpha_{-,\rho}(\eta(x))1\{x\notin G\}.
\]
Then $f_{G,\rho}\in\MF_G$, and
\[
  R_{\phi}(f_{G,\rho})
  \leq
  \int_G C_{\phi}^{+}(\eta(x))dP_X(x)
  +\int_{G^c}C_{\phi}^{-}(\eta(x))dP_X(x)
  +\rho.
\]
The reverse inequality follows pointwise for every $f\in\MF_G$. Letting
$\rho\downarrow0$ gives
\begin{align}
  \mathcal R_{\phi}(G)
  &=
  \int_G C_{\phi}^{+}(\eta(x))dP_X(x)
  +\int_{G^c}C_{\phi}^{-}(\eta(x))dP_X(x)\notag\\
  &=
  \int_{\MX}\Delta C_{\phi}(\eta(x))1\{x\in G\}dP_X(x)
  +\int_{\MX}C_{\phi}^{-}(\eta(x))dP_X(x).
  \label{eq:sequential-set-surrogate-representation}
\end{align}
Classification calibration and Lemma
\ref{lem:conditional-surrogate-risk-sequential}(ii) show that $G_B$
minimizes \eqref{eq:sequential-set-surrogate-representation} over all
measurable prediction sets. Since
$\MF_{\mathrm{all}}=\bigcup_{\substack{G\subseteq\MX\\G\text{ measurable}}}\MF_G$, it follows that
$R_{\phi}^{\ast}=\mathcal R_{\phi}(G_B)$. Comparing any $G$ with $G_B$
therefore yields
\begin{equation}
\label{eq:sequential-surrogate-gap}
  \mathcal R_{\phi}(G)-R_{\phi}^{\ast}
  =
  \int_{G\triangle G_B}q_{\phi}(\eta(x))dP_X(x).
\end{equation}
Moreover, every $f\in\MF_{\mathrm{all}}$ satisfies
\begin{equation}
\label{eq:sequential-classifier-surrogate-lower-bound}
  R_{\phi}(f)-R_{\phi}^{\ast}
  \geq
  \mathcal R_{\phi}(G_f)-R_{\phi}^{\ast}
  =
  \int_{G_f\triangle G_B}q_{\phi}(\eta(x))dP_X(x).
\end{equation}

\begin{lemma}
\label{lem:surrogate-infimum-under-correct-specification}
Under the assumptions of Proposition \ref{prop:MG-constrained misspecification relation}, if $\MF_{\MG}$ is $R$-correctly specified, then
$
  \inf_{f\in\MF_{\MG}}R_{\phi}(f)=R_{\phi}^{\ast}.
$
\end{lemma}

\begin{proof}
    See Supplemental Appendix.
\end{proof}

\paragraph{Proof of Proposition \ref{prop:MG-constrained misspecification relation}}

\medskip
\noindent\emph{(i) implies (ii).}
Let $\{f_m\}_{m\geq1}\subset\MF_{\MG}$ satisfy
\[
  R_{\phi}(f_m)
  \longrightarrow
  \inf_{f\in\MF_{\MG}}R_{\phi}(f).
\]
By (i) and Lemma \ref{lem:surrogate-infimum-under-correct-specification},
$
  R_{\phi}(f_m)\longrightarrow R_{\phi}^{\ast}.
$
Suppose that $R(f_m)$ does not converge to $R^{\ast}$. Since $R(f_m)-R^{\ast}\geq0$, there exist $\varepsilon\in(0,1]$ and a subsequence $\{f_{m_k}\}_{k\geq1}$ such that
$
  R(f_{m_k})-R^{\ast}\geq\varepsilon
$ for every $k$.

Set
$
  A_k:=G_{f_{m_k}}\triangle G_B
$ and 
$
  B_k:=A_k\cap\{|1-2\eta|\geq\varepsilon/2\}.
$
By \eqref{eq:sequential-classification-gap},
\begin{align*}
  \varepsilon
  &\leq
  \int_{A_k}|1-2\eta(x)|dP_X(x)\leq
  \frac{\varepsilon}{2}
  +P_X(B_k).
\end{align*}
Thus
\begin{equation}
\label{eq:sequential-positive-away-boundary-mass}
  P_X(B_k)\geq\frac{\varepsilon}{2}.
\end{equation}
By Lemma \ref{lem:conditional-surrogate-risk-sequential}(iv),
$\kappa_{\phi}(\varepsilon/2)>0$. Combining
\eqref{eq:sequential-classifier-surrogate-lower-bound} and
\eqref{eq:sequential-positive-away-boundary-mass}, we obtain
\begin{align*}
  R_{\phi}(f_{m_k})-R_{\phi}^{\ast}
  &\geq
  \int_{A_k}q_{\phi}(\eta(x))dP_X(x)
  \geq
  \int_{B_k}q_{\phi}(\eta(x))dP_X(x)\\
  &\geq
  \kappa_{\phi}(\varepsilon/2)P_X(B_k)
  \geq
  \kappa_{\phi}(\varepsilon/2)\frac{\varepsilon}{2}>0.
\end{align*}
The lower bound is independent of $k$, contradicting
$R_{\phi}(f_{m_k})\rightarrow R_{\phi}^{\ast}$. Therefore
$R(f_m)\rightarrow R^{\ast}$.

Assume (ii), and set
$
  I_{\phi}:=\inf_{f\in\MF_{\MG}}R_{\phi}(f)
$.
This infimum is finite. Indeed, if
$b=\inf_{|t|\leq1}\phi(t)>-\infty$, then $R_{\phi}(f)\geq b$ for every $f\in\MF_{\MG}$. Since $\MG$ is nonempty, choose $G_0\in\MG$ and define
\[
  f_0(x)
  :=
  \begin{cases}
    0, & x\in G_0,\\
    -1, & x\notin G_0.
  \end{cases}
\]
Then $f_0\in\MF_{\MG}$ and $R_{\phi}(f_0)<\infty$, because $Yf_0(X)\in\{-1,0,1\}$ and $\phi$ is finite-valued. Hence $I_{\phi}\in\Real$.

By the definition of the infimum, there exists a sequence
$\{f_m\}_{m\geq1}\subset\MF_{\MG}$ such that
$
  I_{\phi}
  \leq R_{\phi}(f_m)
  <I_{\phi}+(1/m)
$.
Condition (ii) gives $R(f_m)\rightarrow R^{\ast}$. Since the Bayes risk is a lower bound for every classification risk,
\[
  R^{\ast}
  \leq
  \inf_{f\in\MF_{\MG}}R(f)
  \leq
  \liminf_{m\rightarrow\infty}R(f_m)
  =R^{\ast}.
\]
Thus $\inf_{f\in\MF_{\MG}}R(f)=R^{\ast}$, which is (i).
\hfill\qedsymbol

\bigskip{}

\paragraph{Proof of Theorem \ref{thm:univesal equivalence}.}

\ \\
 (`If' part)

Theorem \ref{thm:risk equivalence} shows for any measurable subsets $G_{1},G_{2}\in 2^{\MX}$
that the surrogate risk ordering 
${\cal R}_{\phi_{1}}\left(G_{1}\right) \leq {\cal R}_{\phi_{1}}\left(G_{2}\right)$
is equivalent to 
\begin{equation*}
\int_{G_{1}\backslash G_{2}}\Delta C_{\phi_{1}}\left(\eta(x)\right)dP_{X}(x) \leq
\int_{G_{2}\backslash G_{1}}\Delta C_{\phi_{1}}\left(\eta(x)\right)dP_{X}(x).
\end{equation*}
This inequality remains true if we replace $\Delta C_{\phi_{1}}\left(\eta(x)\right)$
with $\Delta C_{\phi_{2}}\left(\eta(x)\right)=c\Delta C_{\phi_{1}}\left(\eta(x)\right)$
with $c>0$. Thus 
${\cal R}_{\phi_{1}}\left(G_{1}\right) \leq {\cal R}_{\phi_{1}}\left(G_{2}\right) \Rightarrow
{\cal R}_{\phi_{2}}\left(G_{1}\right) \leq {\cal R}_{\phi_{2}}\left(G_{2}\right)$.
The same argument applies in the opposite direction, hence
$\phi_{1}\overset{u}{\sim}\phi_{2}$.

\ \\
 (`Only if' part)

By Lemma~\ref{lem:conditional-surrogate-risk-sequential}~(i), $\Delta C_{\phi} \left(\eta\right)$ is continuous in $\eta$.
By definition, for any classification-calibrated $\phi$, 
$\Delta C_{\phi} \left(\eta\right) > 0$ for $\eta < 1/2$ and 
$\Delta C_{\phi} \left(\eta\right) < 0$ for $\eta > 1/2$.
Moreover, it follows by Lemma~\ref{lem:conditional-surrogate-risk-sequential}~(ii) that 
$\Delta C_{\phi_{1}} \left(1/2\right) = \Delta C_{\phi_{2}} \left(1/2\right) = 0$.

If $\Delta C_{\phi_{2}}\left(\eta(x)\right)$ is not proportional to 
$\Delta C_{\phi_{1}}\left(\eta(x)\right)$ then we can construct a counterexample
in which $\phi_{1}$ and $\phi_{2}$ rank two prediction sets differently in terms
of surrogate risk and hence $\phi_{1}\overset{u}{\not\sim}\phi_{2}$.

Suppose that 
$\Delta C_{\phi_{2}}\left(\eta(x)\right) \neq c\Delta C_{\phi_{1}}\left(\eta(x)\right)$
on $\eta \in [0,1]$ for any $c>0$.
Then there exist $\eta_{a},\eta_{b} \in [0,1]\backslash\{1/2\}$, $\eta_{a} \neq \eta_{b}$,
for which
\begin{equation}
\frac{\Delta C_{\phi_{1}}\left(\eta_{a}\right)}{\Delta C_{\phi_{1}}\left(\eta_{b}\right)}
\neq
\frac{\Delta C_{\phi_{2}}\left(\eta_{a}\right)}{\Delta C_{\phi_{2}}\left(\eta_{b}\right)}.
\label{eq:unequal ratios}
\end{equation}

Our counterexamples will use ${\MX}=\left\{ a,b\right\}$ with probability distribution
$P_{\alpha}$ characterized by $\eta (a) = \eta_{a}$, $\eta (b) = \eta_{b}$, and
$P\left(X=a\right) = \alpha \in \left[ 0, 1 \right]$ (with $P\left(X=b\right) = 1-\alpha$).
We will find values of $\alpha$ for which the ranking of two subsets of ${\MX}$
will differ between the two surrogate risks.

The surrogate $\phi$-risk of prediction set $G \in 2^{\MX}$ under distribution $P_{\alpha}$ equals 
\begin{align*}
{\MR}_{\phi}\left(G\right) & =
P\left(X=a\right) \left( C_{\phi}^{-}\left(\eta(a)\right) + \Delta C_{\phi}\left(\eta(a)\right)1\left\{ a\in G\right\}\right) \\ 
& + P\left(X=b\right) \left( C_{\phi}^{-}\left(\eta(b)\right) + \Delta C_{\phi}\left(\eta(b)\right)1\left\{ b\in G\right\}\right) \\
& = {\MR}_{\phi}\left(\emptyset\right) 
  + \alpha \Delta C_{\phi}\left(\eta_a\right)1\left\{ a\in G\right\}
  + \left(1-\alpha\right) \Delta C_{\phi}\left(\eta_b\right)1\left\{ b\in G\right\}.
\end{align*}

Case (\romannumeral1): $\eta_{a}<1/2$ and $\eta_{b}<1/2$.
In this case,
$\Delta C_{\phi_{1}}\left(\eta_a\right) > 0$, $\Delta C_{\phi_{2}}\left(\eta_a\right) > 0$,
$\Delta C_{\phi_{1}}\left(\eta_b\right) > 0$ and $\Delta C_{\phi_{2}}\left(\eta_b\right) > 0$.

Consider the subsets $\left\{ a \right\}$ and $\left\{ b \right\}$.
Their $\phi$-risk ranking is determined by the sign of
\begin{equation*}
{\MR}_{\phi}\left(\left\{ a \right\}\right) - {\MR}_{\phi}\left(\left\{ b \right\}\right) =
\alpha \Delta C_{\phi}\left(\eta_a\right) - \left(1-\alpha\right) \Delta C_{\phi}\left(\eta_b\right).
\end{equation*}
${\MR}_{\phi}\left(\left\{ b \right\}\right) > {\MR}_{\phi}\left(\left\{ a \right\}\right)$
for probability distributions $P_{\alpha}$ with 
$\alpha \in \left[ 0, \frac{\Delta C_{\phi}\left(\eta_b\right)}{\Delta C_{\phi}\left(\eta_a\right) + \Delta C_{\phi}\left(\eta_b\right)} \right)$
and ${\MR}_{\phi}\left(\left\{ b \right\}\right) < {\MR}_{\phi}\left(\left\{ a \right\}\right)$
for $\alpha \in \left( \frac{\Delta C_{\phi}\left(\eta_b\right)}{\Delta C_{\phi}\left(\eta_a\right) + \Delta C_{\phi}\left(\eta_b\right)}, 1 \right]$.
Since $\frac{\Delta C_{\phi}\left(\eta_a\right)}{\Delta C_{\phi}\left(\eta_b\right)} > 0$,
the threshold
$\frac{\Delta C_{\phi}\left(\eta_b\right)}{\Delta C_{\phi}\left(\eta_a\right) + \Delta C_{\phi}\left(\eta_b\right)} = \left(1 + \frac{\Delta C_{\phi}\left(\eta_a\right)}{\Delta C_{\phi}\left(\eta_b\right)} \right)^{-1}$ 
where the $\phi$-risk ranking of $G_a$ and $G_b$ reverses
is always inside the $(0,1)$ interval. By (\ref{eq:unequal ratios}), these thresholds will
be different for surrogate loss functions $\phi_{1}$ and $\phi_{2}$. Hence for distributions
with $\alpha$ between the two thresholds, the surrogate $\phi_{1}$- and $\phi_{2}$-risk rankings of
$\left\{ a \right\}$ and $\left\{ b \right\}$ will differ, hence $\phi_{1}\overset{u}{\not\sim}\phi_{2}$.

Case (\romannumeral2): $\eta_{a}>1/2$ and $\eta_{b}>1/2$.
In this case, 
$\Delta C_{\phi_{1}}\left(\eta_a\right) < 0$, $\Delta C_{\phi_{2}}\left(\eta_a\right) < 0$,
$\Delta C_{\phi_{1}}\left(\eta_b\right) < 0$ and $\Delta C_{\phi_{2}}\left(\eta_b\right) < 0$.

Similarly to (\romannumeral1), we compare the ranking of $\left\{ a \right\}$ and $\left\{ b \right\}$
and obtain the same formula for the threshold value of $\alpha$, except that
${\MR}_{\phi}\left(\left\{ b \right\}\right) < {\MR}_{\phi}\left(\left\{ a \right\}\right)$
for  
$\alpha \in \left[ 0, \frac{\Delta C_{\phi}\left(\eta_b\right)}{\Delta C_{\phi}\left(\eta_a\right) + \Delta C_{\phi}\left(\eta_b\right)} \right)$
and ${\MR}_{\phi}\left(\left\{ b \right\}\right) > {\MR}_{\phi}\left(\left\{ a \right\}\right)$
for $\alpha \in \left( \frac{\Delta C_{\phi}\left(\eta_b\right)}{\Delta C_{\phi}\left(\eta_a\right) + \Delta C_{\phi}\left(\eta_b\right)}, 1 \right]$.

Case (\romannumeral3): $\eta_{a}<1/2$ and $\eta_{b}>1/2$.
In this case,
$\Delta C_{\phi_{1}}\left(\eta_a\right) > 0$, $\Delta C_{\phi_{2}}\left(\eta_a\right) > 0$,
$\Delta C_{\phi_{1}}\left(\eta_b\right) < 0$ and $\Delta C_{\phi_{2}}\left(\eta_b\right) < 0$.
${\MR}_{\phi}\left(\left\{ b \right\}\right) < {\MR}_{\phi}\left(\left\{ a \right\}\right)$
for all values of $\alpha \in [0,1]$ both under $\phi_{1}$ and $\phi_{2}$, so we will instead 
compare the surrogate risk ranking of subsets $\left\{a,b \right\}$ and $\emptyset$.

Their $\phi$-risk ranking is determined by the sign of
\begin{equation*}
{\MR}_{\phi}\left(\left\{ a,b \right\}\right) - {\MR}_{\phi}\left(\emptyset\right) =
\alpha \Delta C_{\phi}\left(\eta_a\right) + \left(1-\alpha\right) \Delta C_{\phi}\left(\eta_b\right),
\end{equation*}
which is negative for
$\alpha \in \left[ 0, \left(1 - \frac{\Delta C_{\phi}\left(\eta_a\right)}{\Delta C_{\phi}\left(\eta_b\right)} \right)^{-1} \right)$
and positive for
$\alpha \in \left( \left(1 - \frac{\Delta C_{\phi}\left(\eta_a\right)}{\Delta C_{\phi}\left(\eta_b\right)} \right)^{-1}, 1 \right]$.
The thresholds
$\left(1 - \frac{\Delta C_{\phi}\left(\eta_a\right)}{\Delta C_{\phi}\left(\eta_b\right)} \right)^{-1}$
are interior on $\left( 0,1 \right)$ and different for $\phi_{1}$ and $\phi_{2}$, hence for values
of $\alpha$ between the two thresholds the surrogate $\phi_{1}$- and $\phi_{2}$-risk rankings of  
$\left\{ a,b \right\}$ and $\emptyset$ will disagree. Therefore $\phi_{1}\overset{u}{\not\sim}\phi_{2}$.

Case (\romannumeral4): $\eta_{a}>1/2$ and $\eta_{b}<1/2$ is similar to case (\romannumeral3).
\hfill \qedsymbol 

\bigskip{}

\paragraph{Proof of Theorem \ref{thm:continuous functions}.}

For any classifier $f:\MX \rightarrow [-1,1]$ and for every $x \in \MX$,
\begin{equation*}
f(x) = -1 + \int_{-1}^{1} 1\left\{  f(x) \geq t \right\} dt.
\end{equation*}
This allows us to characterize the hinge risk of $f$ 
through classification risks of its upper level sets as follows: 
\begin{align}
R_h(f)
&= \int_{\MX}\left[\eta(x)(1-f(x))+(1-\eta(x))(1+f(x))\right]dP_{X}(x) \nonumber \\
&= 1 + \int_{\MX}\left(1-2\eta(x)\right)f(x)dP_{X}(x) \nonumber \\
&= 1 + \int_{\MX}\left(1-2\eta(x)\right)
\left[ - 1 + \int_{-1}^{1} 1\left\{ f(x) \geq t \right\} dt \right] dP_{X}(x) \nonumber \\
&= 1 - \int_{\MX}\left(1-2\eta(x)\right) dP_{X}(x)
+ \int_{\MX}\left(1-2\eta(x)\right) 
\left[ \int_{-1}^{1} 1\left\{ x \in G_{f} (t) \right\} dt \right] dP_{X}(x) \nonumber \\
&= 2\int_{\MX} \eta(x) dP_{X}(x)
+ \int_{\MX} \int_{-1}^{1} \left(1-2\eta(x)\right) 1\left\{ x \in G_{f} (t) \right\} dt \ dP_{X}(x) \nonumber \\
&= 2P\left(Y=1\right)
+ \int_{-1}^{1} \int_{\MX}  \left(1-2\eta(x)\right) 1\left\{ x \in G_{f} (t) \right\} dP_{X}(x) \ dt \nonumber \\
&= 2P\left(Y=1\right)
+ \int_{-1}^{1} \left[ \MR\left(G_{f}(t)\right) - P(Y=1) \right]  dt \nonumber \\
&= \int_{-1}^{1} \MR\left(G_{f}(t)\right) dt. \label{eq:Rh layer cake}
\end{align}
The second to last line uses equation (\ref{eq:classification risk at G}), showing that $\MR(G)$ can
be written as 
\begin{equation*}
\MR(G) = P(Y=1) + \int_{\MX}(1-2\eta(x))\cdot1\{x\in G\}dP_{X}(x).
\end{equation*}

For any $f\in\widetilde{\MF}_{\MG}$, $G_{f}(t) \in \MG$ for all $t \in (-1,1]$ by \ref{asm:level set condition}.
Let $\MR^{\ast} = \inf_{G \in \MG} \MR(G)$ be the minimum classification risk over $\MG$,
then $\MR\left(G_{f}(t)\right) \geq \MR^{\ast}$ for all $t \in (-1,1]$,
therefore, the hinge risk of any $f\in\widetilde{\MF}_{\MG}$ is bounded below by $2\MR^{\ast}$:
\begin{equation*}
R_h (f) = \int_{-1}^{1} \MR\left(G_{f}(t)\right) dt \geq 2\MR^{\ast}.
\end{equation*}

Let $G^{\ast} \in \arginf_{G \in \MG} \MR(G)$ be a prediction set minimizing classification risk over $\MG$ (which exists by assumption).
Consider the classifier $\bar{f}_{G^{\ast}} (x) = 2 \cdot 1\left\{x \in G^{\ast}\right\} -1$.
By \ref{asm:optimizer condition}, $\bar{f}_{G^{\ast}} \in \widetilde{\MF}_{\MG}$.
The upper level sets of $\bar{f}_{G^{\ast}}$ are $G_{\bar{f}_{G^{\ast}}}(t) = G^{\ast}$
for all $t \in (-1,1]$ and $\MR\left(G^{\ast}\right)=\MR^{\ast}$,
thus $R_h\left(\bar{f}_{G^{\ast}}\right) = 2\MR^{\ast}$ by (\ref{eq:Rh layer cake}).
Since $R_h(f) \geq 2\MR^{\ast}$ for all $f \in \widetilde{\MF}_{\MG}$,
this proves claim (i) of the theorem.

The last step is to show that $R\left(\tilde{f}^{\ast}\right) = \MR^{\ast}$ for any 
$\tilde{f}^{\ast} \in \arg\inf_{f \in \widetilde{\MF}_{\MG}} R_h (f)$,
i.e., that every minimizer of surrogate risk also minimizes classification risk.

First, we will show that $\MR\left(G_{f}(t)\right)$ is a left-continuous function of $t$.
Take any increasing sequence $t_i \uparrow t$. The classification risk of $G_{f}(t_i)$
can be expressed as the sum of probabilities of two events $A_i$ and $B_i$:
\begin{equation*}
\MR\left(G_{f}(t_i)\right) = P\underbrace{\left(Y=-1,X \in G_{f}(t_i)\right)}_{A_{i}}
 + P\underbrace{\left(Y=1,X \in G^c_{f}(t_i)\right)}_{B_{i}}.
\end{equation*}
Since $G_{f}(t) = \bigcap\limits_{i=1}^{\infty} G_{f}(t_i)$,
$A_{i}$ is a decreasing sequence of events with
$A_{i} \downarrow A \equiv \left\{Y=-1,X \in G_{f}(t)\right\}$
and $B_{i}$ is an increasing sequence of events with
$B_{i} \uparrow B \equiv \left\{Y=1,X \in G^c_{f}(t)\right\}$.
Then by the continuity of measure theorem (\citeauthor{Billingsley_1995} (\citeyear{Billingsley_1995}, Theorem 2.1)),
\begin{equation*}
\lim_{t_i \uparrow t} \MR\left(G_{f}(t_i)\right) = P(A) + P(B) = \MR\left(G_{f}(t)\right).
\end{equation*}

Suppose that $\tilde{f}^{\ast}$ does not minimize classification risk, that is,
$R\left(\tilde{f}^{\ast}\right) > \MR^{\ast}$.
The prediction set of $\tilde{f}^{\ast}$ is $G_{\tilde{f}^{\ast}} (0)$, thus
$\MR\left( G_{\tilde{f}^{\ast}} (0) \right) = R\left(\tilde{f}^{\ast}\right) > \MR^{\ast}$.
Left-continuity of $\MR\left( G_{\tilde{f}^{\ast}} (t) \right)$ at $t=0$ implies that
$\MR\left( G_{\tilde{f}^{\ast}} (t) \right) > \MR^{\ast}$ on an interval $[-\varepsilon,0]$ 
for some $\varepsilon>0$, therefore
\begin{equation*}
R_h\left(\tilde{f}^{\ast}\right) = \int_{-1}^{1} \MR\left(G_{\tilde{f}^{\ast}}(t)\right) dt
> 2\MR^{\ast}.
\end{equation*}
This contradicts the assumption that $\tilde{f}^{\ast}$ minimizes hinge risk over 
$\widetilde{\MF}_{\MG}$, since its minimum has been shown to be $2\MR^{\ast}$.
This proves statement (ii) of the theorem.
\hfill \qedsymbol

\bigskip{}

\bibliographystyle{ecta}
\bibliography{ref_surrogate_loss,EWM}

\newpage

\part*{Supplemental Appendix}
\bigskip{}

\begin{proof}[Proof of Lemma~\ref{lem:conditional-surrogate-risk-sequential}]
Let $A_{+}=[0,1]$ and $A_{-}=[-1,0)$. For $s\in\{+,-\}$, write
\[
  V_s(u):=\inf_{a\in A_s}C_{\phi}(a,u),
\]
so that $V_{+}=C_{\phi}^{+}$ and $V_{-}=C_{\phi}^{-}$. Let
\[
  b:=\inf_{|t|\leq1}\phi(t)>-\infty.
\]
For every $a\in[-1,1]$ and $u\in[0,1]$,
\[
  C_{\phi}(a,u)
  =u\phi(a)+(1-u)\phi(-a)
  \geq b.
\]
For each $s$, choose any $a_s\in A_s$. Since $\phi$ is finite-valued,
\[
  b\leq V_s(u)\leq C_{\phi}(a_s,u)<\infty.
\]
Thus $V_s$ is finite on $[0,1]$.

For fixed $a$, the map $u\mapsto C_{\phi}(a,u)$ is affine. Hence $V_s$, being the infimum of affine functions, is concave and therefore continuous on $(0,1)$. We verify continuity at the endpoints. For every $u\in[0,1]$,
\[
  V_s(u)\geq ub+(1-u)V_s(0),
\]
so $\liminf_{u\downarrow0}V_s(u)\geq V_s(0)$. For any $\varepsilon>0$, choose $a_{\varepsilon}\in A_s$ such that
\[
  \phi(-a_{\varepsilon})\leq V_s(0)+\varepsilon.
\]
Then
\[
  V_s(u)
  \leq u\phi(a_{\varepsilon})+(1-u)\phi(-a_{\varepsilon}),
\]
and therefore $\limsup_{u\downarrow0}V_s(u)\leq V_s(0)+\varepsilon$. Letting $\varepsilon\downarrow0$ proves right-continuity at $0$. Similarly,
\[
  V_s(u)\geq uV_s(1)+(1-u)b,
\]
and an $\varepsilon$-minimizer of $a\mapsto\phi(a)$ over $A_s$ proves left-continuity at $1$. This establishes (i).

Classification calibration gives
\[
  \Delta C_{\phi}(u)>0\quad(u<1/2),
  \qquad
  \Delta C_{\phi}(u)<0\quad(u>1/2).
\]
Continuity of $\Delta C_{\phi}$ therefore implies $\Delta C_{\phi}(1/2)=0$. Part (ii) follows.

By continuity on the compact interval $[0,1]$, $q_{\phi}$ attains a finite maximum $M_{\phi}$. Since $q_{\phi}(0)>0$, this maximum is strictly positive. Moreover, $q_{\phi}(1/2)=0$ and continuity at $1/2$ imply that, for every $\varepsilon>0$, there is $r>0$ such that
\[
  \left|u-\frac12\right|<r
  \quad\Longrightarrow\quad
  q_{\phi}(u)\leq\varepsilon.
\]
Setting $a:=\min\{1,r\}$ proves (iii), because $|1-2u|\leq a$ implies $|u-1/2|\leq a/2<r$.

Finally, for $\delta\in(0,1]$, the set
\[
  K_{\delta}:=\{u\in[0,1]:|1-2u|\geq\delta\}
\]
is nonempty and compact. The continuous function $q_{\phi}$ attains its minimum on $K_{\delta}$, and this minimum is strictly positive by (ii), since $1/2\notin K_{\delta}$. This proves (iv).
\end{proof}

\begin{proof}[Proof of \ref{lem:surrogate-infimum-under-correct-specification}]
Since $\MF_{\MG}=\bigcup_{G\in\MG}\MF_G$,
\begin{equation}
\label{eq:sequential-surrogate-decomposition}
  \inf_{f\in\MF_{\MG}}R_{\phi}(f)
  =
  \inf_{G\in\MG}\mathcal R_{\phi}(G).
\end{equation}
Because $\MF_{\MG}\subseteq\MF_{\mathrm{all}}$, the left-hand side is at least $R_{\phi}^{\ast}$. Suppose that the inequality is strict and set
\[
  d
  :=
  \inf_{f\in\MF_{\MG}}R_{\phi}(f)-R_{\phi}^{\ast}
  >0.
\]
By \eqref{eq:sequential-surrogate-gap} and \eqref{eq:sequential-surrogate-decomposition}, every $G\in\MG$ satisfies
\begin{equation}
\label{eq:sequential-positive-surrogate-gap}
  d
  \leq
  \int_{G\triangle G_B}q_{\phi}(\eta(x))dP_X(x).
\end{equation}
By Lemma \ref{lem:conditional-surrogate-risk-sequential}(iii),
\[
  M_{\phi}:=\max_{u\in[0,1]}q_{\phi}(u)\in(0,\infty),
\]
and there exists $a\in(0,1]$ such that
\[
  |1-2u|\leq a
  \quad\Longrightarrow\quad
  q_{\phi}(u)\leq\frac d2.
\]
Writing $A_G:=G\triangle G_B$ and splitting the integral in \eqref{eq:sequential-positive-surrogate-gap} gives
\begin{align*}
  d
  &\leq
  \int_{A_G\cap\{|1-2\eta|\leq a\}}q_{\phi}(\eta(x))dP_X(x)
  +
  \int_{A_G\cap\{|1-2\eta|>a\}}q_{\phi}(\eta(x))dP_X(x)\\
  &\leq
  \frac d2
  +M_{\phi}P_X\bigl(A_G\cap\{|1-2\eta|>a\}\bigr).
\end{align*}
Consequently,
\[
  P_X\bigl(A_G\cap\{|1-2\eta|>a\}\bigr)
  \geq\frac{d}{2M_{\phi}}
  \qquad(G\in\MG).
\]
Using \eqref{eq:sequential-classification-gap}, we obtain, uniformly over $G\in\MG$,
\[
  \mathcal R(G)-R^{\ast}
  \geq
  aP_X\bigl(A_G\cap\{|1-2\eta|>a\}\bigr)
  \geq
  \frac{ad}{2M_{\phi}}>0.
\]
It follows that
\[
  \inf_{f\in\MF_{\MG}}R(f)
  =\inf_{G\in\MG}\mathcal R(G)
  >R^{\ast},
\]
contradicting $R$-correct specification. Hence the strict inequality is impossible.
\end{proof}

\section{Proof of the results in Sections \ref{sec:statistical property} and \ref{sec:Applications to monotone classification} (comprehensive version)}
\label{appx:proof 2}

This appendix provides proofs of the results in Sections \ref{sec:statistical property} and \ref{sec:Applications to monotone classification}. Given $G \in \MG$, let $\mathbf{1}_G$ be an indicator function on $\MX$ such that $\mathbf{1}_G(x)=1\{x \in G\}$.  We first give the definition of bracketing entropy of $\MG$.

\bigskip{}

\begin{definition}[Bracketing entropy]\label{def:bracketing entropy}
(i) Let $\mathcal{\MF}$ be a class of functions on $\MX$. For $f\in\MF$, let $\left\Vert f\right\Vert _{p,Q}{\equiv}\left(\int_{\MX}\left|f(x)\right|^{p}dQ(x)\right)^{1/p}$. $\left\Vert \cdot\right\Vert _{p,Q}$ is the $L_{p}\left(Q\right)$-metric on $\MX$, where $Q$ is a measure on $\MX$. 
Given a pair of functions $\left(f_{1},f_{2}\right)$ with $f_1 \leq f_2$, let $[f_1,f_2]{\equiv}\{f \in \MF:f_1 \leq f \leq f_2\}$ be the bracket.   
Given $\epsilon>0$, let $N_{p}^{B}\left(\epsilon,\MF,Q\right)$ be the smallest $k$ such that there exist pairs of functions $\left(f_{j}^{L},f_{j}^{U}\right)$, $j=1,\ldots,k$,
with $f_{j}^{L}\leq f_{j}^{U}$ that satisfy $\left\Vert f_{j}^{U}-f_{j}^{L}\right\Vert _{p,Q}<\epsilon$ and 
\begin{align*}
\MF \subseteq  \cup_{j=1}^{k}\left[f_{j}^{L},f_{j}^{U}\right].
\end{align*}
We term $N_{p}^{B}\left(\epsilon,\MF,Q\right)$ the $L_{p}\left(Q\right)$-bracketing number of $\MF$, and $H_{p}^{B}\left(\epsilon,\MF,Q\right) \equiv \log N_{p}^{B}\left(\epsilon,\MF,Q\right)$ the $L_{p}\left(Q\right)$-bracketing entropy of $\MF$.
We also refer to $[f_{j}^{L},f_{j}^{U}]$ as the $\epsilon$-bracket with respect to $L_{p}(Q)$ if and only if $\left\Vert f_{j}^{U}-f_{j}^{L}\right\Vert _{p,Q}<\epsilon$ holds. \\  
(ii) Given a class of measurable subsets $\MG \subseteq 2^{\MX}$, let $\MH_{\MG}\equiv\{\mathbf{1}_G:G \in \MG\}$. We define $H_{p}^{B}\left(\epsilon,\MG,Q\right) \equiv H_{p}^{B}\left(\epsilon,\MH_{\MG},Q\right)$ and term this the $L_{p}\left(Q\right)$-bracketing entropy of $\MG$.
\end{definition}
\bigskip{}

We follow \cite{MammenTsybakov1999} and \cite{MT17} for the general strategy of the proof of Theorem \ref{thm:statistical propery for excess classification risk}. Before presenting the proof of Theorem \ref{thm:statistical propery for excess classification risk}, we introduce two auxiliary lemmas.
The following lemma is equivalent to Corollary D.1 in \cite{MT17}.

\bigskip{}
\begin{lemma} \label{lem:empirical_process_result} (Corollary D.1 in \cite{MT17}) Let $\{Z_i\}_{i=1}^{n}$ be a sequence of i.i.d. random vectors with distribution $P$.
Let $Z=\left(Z_{1},Z_{2}\right)\sim P$, and $\MF$ be a class of real-valued functions of the form
$f\left(z\right)=f\left(z_{1},z_{2}\right)=g\left(z\right)\cdot h\left(z_{2}\right)$,
where $h\in\mathcal{H}$, $\mathcal{H}$ is a class of functions with
values in $\{0,1\}$, and $g$ is some fixed real-valued function (which may depend on $P$) such that $|g| \leq 1$. Suppose $\mathcal{H}$ satisfies 
\begin{align*}
H_{1}^{B}\left(\epsilon,\mathcal{H},P_{2}\right)\leq K\epsilon^{-r}
\end{align*}
for some constants $K>0$ and $r \geq1$ and for all $0<\epsilon\leq1$. Suppose also that $P_2$, the marginal distribution of $Z_2$, has a density with respect to Lebesgue measure bounded above by $A$. 
Then there exist positive constants $D_{1,}D_{2},D_{3}$ (which depend only on
$K$, $r$, and $A$) such that, for $n\geq3$,
\begin{align*}
P^{n}\left(\sup_{f\in\MF}\left|\frac{1}{\sqrt{n}}\sum_{i=1}^{n}\left(f\left(Z_{i}\right)-E_P[f\left(Z_{i}\right)]\right)\right|>xq_{n}\right)\leq  D_{2}\exp\left(-x^{2}q_{n}^{2}\right),
\end{align*}
for $D_{1}\leq x\leq D_{3}\sqrt{n}/q_{n}$, where 
\begin{align*}
q_{n}= \begin{cases}
\begin{array}{c}
\log n\\
n^{\left(r-1\right)/2\left(r+1\right)}
\end{array} \begin{array}{l}
r=1\\
r>1
\end{array}\end{cases}.
\end{align*}
\end{lemma}

\bigskip{}

The following lemma plays the same role as Lemma D.1 in \cite{MT17}.

\bigskip{}

\begin{lemma}\label{lem:empirical process result_2}
Maintain the same definitions and assumptions as in Lemma \ref{lem:empirical_process_result} with $r \geq 1$ replaced by $0< r<1$.
Then there exist positive constants $D_3$ and $D_4$, depending only on $K$ and $r$, such that: 
\begin{align*}
E_{P^{n}}\left[\sup_{f\in\MF}\left|\frac{1}{n}\sum_{i=1}^{n}f(Z_{i})-E_{P}\left[f(Z)\right]\right|\right]\leq  \frac{D_{3}}{\sqrt{n}}+\frac{D_{4}}{n}.
\end{align*}
\end{lemma}

\begin{proof}
We apply Proposition~3.5.15 in \cite{Gine_Nickl_2016}.
Note first that $\|f\|_{1,P}\leq 1$ and $\|f\|_{2,P}\leq 1$ for all $f\in\MF$. We can therefore apply Proposition~3.5.15 in \cite{Gine_Nickl_2016}, with $F=1$ and $\delta=1$, to obtain
\begin{align}
&E_{P^{n}}\left[
\sup_{f\in\MF}
\left|
\frac{1}{n}\sum_{i=1}^{n}f(Z_i)-E_P[f(Z)]
\right|
\right]
\notag\\
&\leq
\left(
\frac{58}{\sqrt n}
+
\frac{1}{3n}
\int_0^2
\sqrt{\log\left(2N_2^B(\epsilon,\MF,P)\right)}
\,d\epsilon
\right)
\int_0^2
\sqrt{\log\left(2N_2^B(\epsilon,\MF,P)\right)}
\,d\epsilon
\notag\\
&\leq
\left(
\frac{58}{\sqrt n}
+\frac{2}{3n}
+\frac{1}{3n}
\int_0^2
\sqrt{H_2^B(\epsilon,\MF,P)}
\,d\epsilon
\right)
\left(
2+
\int_0^2
\sqrt{H_2^B(\epsilon,\MF,P)}
\,d\epsilon
\right).
\label{eq:Gini_Nickl_2016_bound}
\end{align}

We next bound the bracketing entropy of the product class $\MF$.
Fix $0<\epsilon\leq 1$, and let $[h_j^L,h_j^U]$, $j=1,\ldots,N_1^B(\epsilon^2,\MH,P_2)$, be a collection of $\epsilon^2$-brackets of $\MH$ with respect
to $L_1(P_2)$ such that
\[
\|h_j^U-h_j^L\|_{1,P_2}<\epsilon^2
\]
and
\[
\MH
\subseteq
\bigcup_{j=1}^{N_1^B(\epsilon^2,\MH,P_2)}
[h_j^L,h_j^U].
\]
Because every $h\in\MH$ takes values in $\{0,1\}$, the bracket
endpoints can be truncated to $[0,1]$ without increasing their
widths. We may therefore assume that $0\leq h_j^L\leq h_j^U\leq 1$.

Define $g^+(z):=\max\{g(z),0\}$ and $g^-(z):=\max\{-g(z),0\}$ so that $g=g^+-g^-$. For each $j$, define
\begin{align*}
f_j^L(z)
&:=
g^+(z)h_j^L(z_2)-g^-(z)h_j^U(z_2),\\
f_j^U(z)
&:=
g^+(z)h_j^U(z_2)-g^-(z)h_j^L(z_2).
\end{align*}
If $h_j^L\leq h\leq h_j^U$, then
\[
f_j^L(z)
\leq
g(z)h(z_2)
\leq
f_j^U(z)
\]
for every $z$, regardless of the sign of $g(z)$. Moreover,
\[
f_j^U(z)-f_j^L(z)
=
|g(z)|\{h_j^U(z_2)-h_j^L(z_2)\}.
\]
Since $|g|\leq 1$ and $0\leq h_j^U-h_j^L\leq 1$, it follows that
\begin{align*}
\|f_j^U-f_j^L\|_{2,P}^2
&=
E_P\left[
|g(Z)|^2
\{h_j^U(Z_2)-h_j^L(Z_2)\}^2
\right]\\
&\leq
E_{P_2}\left[
\{h_j^U(Z_2)-h_j^L(Z_2)\}^2
\right]\\
&\leq
E_{P_2}\left[
h_j^U(Z_2)-h_j^L(Z_2)
\right]\\
&=
\|h_j^U-h_j^L\|_{1,P_2}
<\epsilon^2.
\end{align*}
Thus, $[f_j^L,f_j^U]$ is an $\epsilon$-bracket for $\MF$ with
respect to $L_2(P)$. Consequently,
\begin{align}
N_2^B(\epsilon,\MF,P)
&\leq
N_1^B(\epsilon^2,\MH,P_2),\notag\\
H_2^B(\epsilon,\MF,P)
&\leq
H_1^B(\epsilon^2,\MH,P_2)
\leq
K\epsilon^{-2r}.
\label{eq:product_bracketing_bound}
\end{align}

For $\epsilon>1$, the single bracket $[-g^-,g^+]$ contains every $f\in\MF$, and its $L_2(P)$ width is at most one.
Hence,
\[
H_2^B(\epsilon,\MF,P)=0
\qquad
\text{for }\epsilon>1.
\]
It follows from \eqref{eq:product_bracketing_bound} that
\begin{align*}
\int_0^2
\sqrt{H_2^B(\epsilon,\MF,P)}
\,d\epsilon
&\leq
\int_0^1
\sqrt K\,\epsilon^{-r}
\,d\epsilon=
\frac{\sqrt K}{1-r}\leq
\frac{2^{1-r}\sqrt K}{1-r}.
\end{align*}

Substituting this bound into
\eqref{eq:Gini_Nickl_2016_bound} yields
\begin{align*}
&E_{P^{n}}\left[
\sup_{f\in\MF}
\left|
\frac{1}{n}\sum_{i=1}^{n}f(Z_i)-E_P[f(Z)]
\right|
\right]\leq
\left(
\frac{58}{\sqrt n}
+\frac{2}{3n}
+\frac{2^{1-r}\sqrt K}{3n(1-r)}
\right)
\left(
2+\frac{2^{1-r}\sqrt K}{1-r}
\right).
\end{align*}
Therefore, setting
\begin{align*}
D_3
=
116+
\frac{29\cdot 2^{2-r}\sqrt K}{1-r}, \quad D_4
=
\frac{1}{3}
\left(
2+\frac{2^{1-r}\sqrt K}{1-r}
\right)^2
\end{align*}
gives
\[
E_{P^{n}}\left[
\sup_{f\in\MF}
\left|
\frac{1}{n}\sum_{i=1}^{n}f(Z_i)-E_P[f(Z)]
\right|
\right]
\leq
\frac{D_3}{\sqrt n}
+
\frac{D_4}{n},
\]
which proves the result.
\end{proof}

\bigskip{}

We are now ready to present the proof of Theorem \ref{thm:statistical propery for excess classification risk}.

\paragraph{Proof of Theorem
\ref{thm:statistical propery for excess classification risk}}

The empirical classification-risk minimum over $\MF_{\MG}$ is attained,
because only finitely many label assignments are possible on a finite
sample. If $\sum_{i=1}^{n}\omega_i>0$, applying Corollary
\ref{thm:continuous functions_WC} to the empirical distribution of
$\{Z_i=(\omega_i,Y_i,X_i):i=1,\ldots,n\}$ shows that $\hat f$
minimizes $\widehat R^\omega(\cdot)$ over $\MF_{\MG}$.
If $\sum_{i=1}^{n}\omega_i=0$, the same conclusion holds trivially,
since every empirical weighted classification risk is zero.
Thus,
\[
\widehat R^\omega(\hat f)
\leq \widehat R^\omega(f)
\qquad\text{for every }f\in\MF_{\MG}.
\]
We follow the general empirical-process approach used in the proofs
of Theorem~1 in \cite{MammenTsybakov1999} and Proposition~B.1
in \cite{MT17}.

For any bounded measurable function $h$, write
\[
\mathbb P_n h:=\frac{1}{n}\sum_{i=1}^{n}h(Z_i),
\qquad
Ph:=E_P[h(Z)].
\]
Define
\[
\bar g(z):=\frac{\omega y}{M},
\qquad
h_G(z):=\bar g(z)\mathbf{1}\{x\in G\},
\qquad
z=(\omega,y,x),
\]
and let $\MF_0:=\{h_G:G\in\MG\}$.
The bounded-weight assumption implies that
$|\bar g|\leq1$ $P$-almost surely, uniformly over $P\in\MP$.

Let
\[
\bar R^\omega(f):=\frac{1}{M}R^\omega(f),
\qquad
\bar R_n^\omega(f):=\frac{1}{M}\widehat R^\omega(f),
\qquad
\widehat G:=G_{\hat f}.
\]
For every $f\in\MF_{\MG}$,
\begin{align*}
\bar R^\omega(f)
&=
E_P\left[\frac{\omega}{M}\mathbf{1}\{Y=1\}\right]
-Ph_{G_f},\\
\bar R_n^\omega(f)
&=
\frac{1}{n}\sum_{i=1}^{n}
\frac{\omega_i}{M}\mathbf{1}\{Y_i=1\}
-\mathbb P_n h_{G_f}.
\end{align*}
Consequently, empirical risk minimization implies
\[
\mathbb P_n h_{\widehat G}
\geq \mathbb P_n h_{G_f}
\qquad\text{for every }f\in\MF_{\MG}.
\]
Hence, for every $f\in\MF_{\MG}$,
\begin{align*}
\bar R^\omega(\hat f)-\bar R^\omega(f)
&=Ph_{G_f}-Ph_{\widehat G}\\
&=(P-\mathbb P_n)h_{G_f}
  +\bigl(\mathbb P_n h_{G_f}
         -\mathbb P_n h_{\widehat G}\bigr)
  +(\mathbb P_n-P)h_{\widehat G}\\
&\leq
2\sup_{G\in\MG}
\left|(\mathbb P_n-P)h_G\right|.
\end{align*}
Since this bound holds for every $f\in\MF_{\MG}$, it follows that
\begin{align*}
0
\leq
R^\omega(\hat f)-\inf_{f\in\MF_{\MG}}R^\omega(f)
\leq
2M\sup_{G\in\MG}
\left|(\mathbb P_n-P)h_G\right|.
\end{align*}
This argument does not require the population classification-risk
infimum to be attained.

For $n\geq3$, define
\[
q_n:=
\begin{cases}
1, & 0<r<1,\\
\log n, & r=1,\\
n^{(r-1)/(2(r+1))}, & r>1,
\end{cases}
\qquad
S_n:=
\frac{\sqrt n}{q_n}
\sup_{G\in\MG}
\left|(\mathbb P_n-P)h_G\right|.
\]
The preceding inequality becomes
\begin{align*}
0
\leq
R^\omega(\hat f)-\inf_{f\in\MF_{\MG}}R^\omega(f)
\leq
\frac{2Mq_n}{\sqrt n}S_n.
\end{align*}
It therefore suffices to show that
\[
\sup_{P\in\MP}E_{P^n}[S_n]=O(1).
\]

The class $\MF_0$ consists of functions of the form
$\bar g(z)\mathbf{1}\{x\in G\}$, with $|\bar g|\leq1$.
The class of indicator functions
$\{\mathbf{1}_G:G\in\MG\}$ satisfies the bracketing-entropy
condition in the theorem. Since both this entropy bound and the
density bound hold uniformly over $P\in\MP$, the constants in
Lemmas \ref{lem:empirical process result_2} and
\ref{lem:empirical_process_result} can be chosen independently
of $P\in\MP$.

First, suppose that $0<r<1$.
Applying Lemma \ref{lem:empirical process result_2} to $\MF_0$
gives constants $C_1,C_2>0$, independent of $P\in\MP$, such that
\[
E_{P^n}\left[
\sup_{G\in\MG}
\left|(\mathbb P_n-P)h_G\right|
\right]
\leq
\frac{C_1}{\sqrt n}+\frac{C_2}{n}.
\]
Since $q_n=1$ in this case,
\[
\sup_{P\in\MP}E_{P^n}[S_n]
\leq C_1+\frac{C_2}{\sqrt n}
=O(1).
\]

Next, suppose that $r\geq1$.
Write
\[
\bar g_+(z):=\max\{\bar g(z),0\},
\qquad
\bar g_-(z):=\max\{-\bar g(z),0\},
\]
so that $\bar g=\bar g_+-\bar g_-$ and
$0\leq\bar g_+,\bar g_-\leq1$.
Consider the two classes
\[
\MF_+
:=
\{z\mapsto\bar g_+(z)\mathbf{1}\{x\in G\}:G\in\MG\},
\qquad
\MF_-
:=
\{z\mapsto\bar g_-(z)\mathbf{1}\{x\in G\}:G\in\MG\}.
\]
Both classes satisfy the conditions of
Lemma \ref{lem:empirical_process_result}.
Let $D_1,D_2,D_3>0$ be constants in that lemma that apply to both
classes and are independent of $P\in\MP$.

Choose a fixed $D\geq2D_1$.
Since $\sqrt n/q_n\to\infty$, for all sufficiently large $n$,
\[
D_1\leq \frac{D}{2}
\leq D_3\frac{\sqrt n}{q_n}.
\]
Applying Lemma \ref{lem:empirical_process_result} separately to
$\MF_+$ and $\MF_-$, and using the triangle inequality and
the union bound, yields
\[
P^n(S_n>D)
\leq
2D_2\exp\left(-\frac{D^2q_n^2}{4}\right)
\]
uniformly over $P\in\MP$ for all sufficiently large $n$.

Moreover, $|h_G|\leq1$ implies
\[
0\leq S_n\leq\frac{2\sqrt n}{q_n}.
\]
Therefore,
\begin{align*}
E_{P^n}[S_n]
&\leq
D+\frac{2\sqrt n}{q_n}P^n(S_n>D)\\
&\leq
D+\frac{4D_2\sqrt n}{q_n}
\exp\left(-\frac{D^2q_n^2}{4}\right).
\end{align*}
For $r\geq1$, we have $q_n^2/\log n\to\infty$, so the second
term converges to zero. Since the bound is independent of
$P\in\MP$, this proves
\[
\sup_{P\in\MP}E_{P^n}[S_n]=O(1).
\]

Combining the two cases, we conclude that
\begin{align*}
\sup_{P\in\MP}
E_{P^n}\left[
R^\omega(\hat f)-\inf_{f\in\MF_{\MG}}R^\omega(f)
\right]
&\leq
\frac{2Mq_n}{\sqrt n}
\sup_{P\in\MP}E_{P^n}[S_n]\\
&=
O\left(\frac{q_n}{\sqrt n}\right)
=
O\left(\tau(n)\right).
\end{align*}
\hfill\qedsymbol

\bigskip{}

We finnaly provide the proof of Lemma \ref{lemma:bracketing entropy_monotone G}.

\bigskip{}

\paragraph{Proof of Lemma
\ref{lemma:bracketing entropy_monotone G}.}
Let $\mu_X$ be the Lebesgue measure on $\MX$.
Since complementation preserves bracketing entropy, the results for
lower layers in \citet{Dudley1999} also apply to $\MG_M$.

Suppose first that $d_x\geq2$. By Theorem 8.3.2 in
\citet{Dudley1999}, there exists a positive constant $K$ such that
\[
H_{1}^{B}(\epsilon,\MG_M,\mu_X)
\leq K\epsilon^{1-d_x}.
\]
When $d_x=1$, Theorem 8.3.1 in \citet{Dudley1999} implies that, for some positive constant $K^\prime$,
\[
H_{1}^{B}(\epsilon,\MG_M,\mu_X)
\leq
K^\prime\{1+\log(1/\epsilon)\}
\leq
2K^\prime\epsilon^{-1/2}
\]
where the last inequality follows from
$1+\log(1/\epsilon)\leq2\epsilon^{-1/2}$ for $0<\epsilon\leq1$.

Since $P_X$ is absolutely continuous with respect to $\mu_X$ and
has a density bounded from above by $A$, every
$(\epsilon/A)$-bracket with respect to $\mu_X$ is an
$\epsilon$-bracket with respect to $P_X$. Hence,
\[
H_{1}^{B}(\epsilon,\MG_M,P_X)
\leq
H_{1}^{B}(\epsilon/A,\MG_M,\mu_X).
\]
Therefore, when $d_x\geq2$,
\[
H_{1}^{B}(\epsilon,\MG_M,P_X)
\leq
K(\epsilon/A)^{1-d_x}
=
KA^{d_x-1}\epsilon^{1-d_x},
\]
whereas, when $d_x=1$,
\[
H_{1}^{B}(\epsilon,\MG_M,P_X)
\leq
2K^\prime(\epsilon/A)^{-1/2}
=
2K^\prime A^{1/2}\epsilon^{-1/2}.
\]
The result follows by setting
$C=KA^{d_x-1}$ and $C^\prime=2K^\prime A^{1/2}$.
\hfill\qedsymbol

\end{document}